\documentclass[10pt,journal,compsoc]{IEEEtran}
\ifCLASSOPTIONcompsoc

\ifCLASSINFOpdf
\else
\fi
\usepackage[utf8]{inputenc}
\usepackage[T1]{fontenc}
\usepackage{url}
\usepackage{ifthen}
\usepackage{cite}
\usepackage[cmex10]{amsmath} 
\usepackage{stackengine}
\usepackage{gensymb}
\usepackage{graphics}
\usepackage{mypackage}
\usepackage{mathrsfs}

\usepackage{mathtools}
\usepackage{tabu}
\usepackage{float}
\usepackage{physics}
\usepackage{siunitx}
\usepackage{multicol}
\usepackage{amssymb}
\usepackage[nomain,acronym,shortcuts]{glossaries}

\usepackage{etoolbox}

\usepackage{enumitem}

\newcommand{\sir}{\mathrm{SIR}}
\newcommand{\Pb}{\mathbb{P}}
\newcommand{\Eb}{\mathbb{E}}

\newcommand{\Lc}{\mathcal{L}}

\def\delequal{\mathrel{\ensurestackMath{\stackon[1pt]{=}{\scriptscriptstyle\Delta}}}}

 \makeglossaries
\newcommand*{\acro}[3][]{\newacronym[#1]{#2}{#2}{#3}}

\newtheorem{remark}{Remark}
\newtheorem{proposition}{Proposition}
\newtheorem{corollary}{Corollary}
\newtheorem{theorem}{Theorem}
\theoremstyle{approximation}

\newcommand{\black}{\textcolor{black}}

\newcommand{\Eq}[1]{(\ref{eq:#1})}
\newcommand{\App}[1]{Appendix~\ref{app:#1}}
\newcommand{\Fig}[1]{Fig.~\ref{fig:#1}}

\DeclareMathOperator*{\argmin}{\arg\!\min}
\DeclareMathOperator*{\R}{\mathbb{R}}

\def\delequal{\mathrel{\ensurestackMath{\stackon[1pt]{=}{\scriptstyle\Delta}}}}

\let\mybibitem\bibitem
\renewcommand{\bibitem}[1]{%
  \ifstrequal{#1}{nature}
    {\color{blue}\mybibitem{#1}}
    {\color{black}\mybibitem{#1}}%
}

\acro{OFDM}{orthogonal frequency-division multiplexing}
\acro{OFDMA}{orthogonal frequency-division multiple access}
\acro{MNO}{mobile network operator}
\acro{RA}{resource allocation}
\acro{SC-FDMA}{single carrier frequency division multiple access}
\acro{CR}{cognitive radio}
\acro{RFIC}{radio frequency integrated circuit} 
\acro{SDR}{software defined radio}
\acro{SDN}{software defined networking}
\acro{su}{secondary user}
\acro{DL}{downlink}
\acro{UL}{uplink}
\acro{QoS}{quality-of-service}
\acro{USRP}{universal software radio peripheral}
\acro{GSM}{Global System for Mobile Communications}
\acro{TDMA}{time-division multiple access}
\acro{FDMA}{frequency-division multiple access}
\acro{GPRS}{General Packet Radio Service}
\acro{MSC}{Mobile Switching Centre}
\acro{BSC}{Base Station Controller}
\acro{UMTS}{Universal Mobile Telecommunications System}
\acro{WCDMA}{wide-band code division multiple access}
\acro{CDMA}{code division multiple access}
\acro{LTE}{Long Term Evolution}
\acro{PAPR}{peak-to-average power rating}
\acro{HetNet}{heterogeneous networks}
\acro{PHY}{physical layer}
\acro{MAC}{medium access control}
\acro{AMC}{adaptive modulation and coding}
\acro{MIMO}{multiple-input multiple-output}
\acro{M-MIMO}{massive MIMO}
\acro{RAT}{radio access technology}
\acro{VNI}{visual networking index}
\acro{RB}{resource block}
\acro{UE}{user equipment}
\acro{CQI}{channel quality indicator}
\acro{HD}{half duplex}
\acro{IBFD}{in-band full duplex}
\acro{SIC}{self-interference cancellation}
\acro{SI}{self-interference}
\acro{BS}{base station}
\acro{FBMC}{filter bank multi-carrier}
\acro{UFMC}{universal filtered multi-carrier}
\acro{SCM}{single carrier modulation}
\acro{isi}{inter-symbol interference}
\acro{FTN}{faster-than-nyquist}
\acro{M2M}{machine-to-machine}
\acro{MTC}{machine type communication}
\acro{mmWave}{millimeter wave}
\acro{BF}{beamforming}
\acro{LoS}{line-of-sight}
\acro{NLoS}{non-line-of-sight}
\acro{CAPEX}{capital expenditure}
\acro{OPEX}{operational expenditure}
\acro{ICT}{information and communications technology}
\acro{SP}{service providers}
\acro{InP}{infrastructure providers}
\acro{MVNP}{mobile virtual network provider}
\acro{MVNO}{mobile virtual network operator}
\acro{NFV}{network function virtualization}
\acro{VNF}{virtual network functions}
\acro{C-RAN}{cloud radio access network}
\acro{RAN}{radio access network}
\acro{BBU}{baseband unit}
\acro{RRH}{remote radio head}
\acro{TDD}{time-division duplexing}
\acro{FDD}{frequency-division duplexing}
\acro{GFDM}{generalized frequency division multiplexing}
\acro{CSI}{channel state information}
\acro{FFT}{fast Fourier transform}
\acro{IFFT}{inverse FFT}
\acro{CFO}{carrier frequency offset}
\acro{CoMP}{coordinated multipoint}
\acro{D2D}{device-to-device}
\acro{OOB}{out-of-band}
\acro{TTI}{transmission time interval}
\acro{DUE}{D2D user equipment}
\acro{DAS}{distributed antenna system}
\acro{ICIC}{inter-cell interference coordination}
\acro{ICI}{inter-cell interference}
\acro{ISI}{inter-symbol interference}
\acro{CP}{cyclic prefix}
\acro{PDF}{probability distribution function}
\acro{KPI}{key performance indicator}
\acro{SBS}{small base station}
\acro{MBS}{macro base station}
\acro{SCN}{small cell network}
\acro{FIFO}{first in first out}
\acro{VCC}{virtual cache center}
\acro{UAV}{unmanned aerial vehicles}
\acro{MPSQ}{multiclass processor sharing queue} 
\acro{EE}{energy efficiency}
\acro{SIR}{signal-to-interference ratio}
\acro{SINR}{signal-to-noise-plus-interference ratio}
\acro{PPP}{Poisson point process}
\acro{PCP}{Poisson cluster process}
\acro{HCP}{hard-core placement} 
\acro{TCP}{Thomas cluster process }
\acro{CPF}{caching popular files}
\acro{GCA}{greedy caching algorithm }
\acro{RC}{random caching }
\acro{PC}{probabilistic caching}
\acro{5G}{fifth generation}
\acro{MEC}{mobile edge computing}
\acro{AP}{access point}
\acro{VoD}{video-on-demand}
\acro{EPC}{evolved packet core}
\acro{QoE}{quality-of-experience}
\acro{CDN}{content delivery networks}
\acro{F-RAN}{fog-radio access network}
\acro{AR}{augmented reality}
\acro{VR}{virtual reality}
\acro{4C}{computing, caching, communication, and control}
\acro{ABR}{Adaptive BitRate}
\acro{ILP}{Integer Linear Program}
\acro{MILP}{Mixed Integer Linear Program}
\acro{MINLP}{Mixed Integer Non-Linear Program}
\acro{BC}{broadcast channel}
\acro{MDS}{maximum distance separable}
\acro{RS}{Reed-Solomon}
\acro{NC}{network coding}
\acro{MSR}{minimum storage regenerating}
\acro{Coop-MIMO}{cooperative \ac{mimo}}
\acro{UDN}{ultra-dense networks}
\acro{ETSI}{European Telecommunications Standards Institute}
\acro{MCP}{Matern cluster process}
\acro{MD-CoMP}{macrodiversity CoMP transmission}
\acro{JT-CoMP}{joint transmission CoMP}
\acro{CoMP-JT}{coordinated multipoint joint transmission}
\acro{MDSD}{multiple devices to the single device}
\acro{PMF}{probability mass function}
\acro{RV}{random variable}
\acro{i.i.d.}{independently and identically distributed}
\acro{MBMS}{multimedia broadcasting multicasting service}
\acro{CCDF}{complementary cumulative distribution function}
\acro{CDF}{cumulative distribution function}
\acro{PGFL}{probability generating functional}
\acro{KKT}{Karush-Kuhn-Tucker}
\acro{PGF}{point generating function}
\acro{BPP}{binomial point process}
\acro{3GPP}{3rd Generation Partnership Project}
\acro{FAA}{Federal Aviation Administration}
\acro{UAS}{unmanned aircraft systems}
\acro{IoT}{Internet-of-things}
\acro{CNPC}{control and non-payload communication}
\acro{FHD}{full-high-definition}
\acro{MGF}{moment generating function}
\acro{3D}{three-dimensional}
\acro{2D}{two-dimensional}
\acro{1D}{one-dimensional}
\acro{CB}{conjugate beamforming}
\acro{AU}{aerial user}
\acro{GU}{ground user}
\acro{CLT}{central limit theorem}
\acro{SE}{spectral efficiency}
\acro{MRT}{maximum ratio transmission}
\acro{ZFBF}{zero-forcing beamforming}
\acro{GBS}{ground base station}
\acro{C&C}{command and control}
\acro{ASE}{area spectral efficiency}
\acro{RWP}{random waypoint}
\acro{AGL}{above ground level}
\acro{w.r.t.}{with respect to}
\acro{GTA}{ground-to-air}
\acro{AUE}{aerial user equipment}
\acro{GUE}{ground user equipment}
\acro{UB}{upper bound}
\acro{LB}{lower bound}
\acro{ABS}{aerial base station}
\acro{UAV-UE}{UAV-\ac{UE}}
\acro{SCDP}{successful content delivery probability}
\acro{HARP}{highest average received power}
\acro{ULA}{uniform linear array}
\acro{GPP}{Gaussian Poisson process}
\acro{NSD}{nearest serving device}
\acro{NCP}{nearest content provider}
\acro{RSD}{randomly serving device}
\acro{RSCP}{randomly-selected content provider}
\acro{BCD}{block coordinate descent}
\acro{NOMA}{non-orthogonal multiple access}
\acro{MPC}{most popular content}
\acro{LPC}{least popular content}
\acro{AI}{artificial intelligence}
\acro{ML}{machine learning}
\acro{TPP}{temporal point process}
\acro{PP}{point process}
\acro{DPP}{Determinantal point process}
\acro{DPPL}{DPP-based learning}
\acro{DRL}{deep reinforcement learning}
\acro{RL}{reinforcement learning}
\acro{RNN}{recurrent neural network}
\acro{RMTPP}{recurrent marked TPP}
\acro{NR}{New Radio}
\acro{HO}{handover}
\acro{HOF}{handover failure}
\acro{RLF}{radio link failure}
\acro{RSRP}{reference signal received power}

\acro{4G}{fourth generation}
\acro{6G}{sixth generation}
\acro{TTT}{time-to-trigger}
\acro{CF}{cell-free}
\acro{CPU}{central processing unit}

\begin{document}
\title{Performance Analysis and Optimization of Cache-Assisted CoMP for Clustered D2D Networks}

\author{Ramy Amer,~\IEEEmembership{Student~Member,~IEEE,} Hesham~ElSawy,~\IEEEmembership{Senior~Member,~IEEE,} Jacek~Kibi\l{}da,~\IEEEmembership{Member,~IEEE,} M.~Majid~Butt,~\IEEEmembership{Senior~Member,~IEEE,} 
~and~Nicola~Marchetti,~\IEEEmembership{Senior~Member,~IEEE}%
 \thanks{The material in this paper is published in part to IEEE WCNC 2019 \cite{8886101}.}
\thanks{Ramy Amer and Jacek Kibi\l{}da and Nicola~Marchetti are with CONNECT Centre for Future Networks, Trinity College Dublin, Ireland. Email:\{ramyr, kibildj, nicola.marchetti\}@tcd.ie.}
\thanks{Hesham ElSawy is with King Fahd University of Petroleum and Minerals (KFUPM), Saudi Arabia. Email: hesham.elsawy@kfupm.edu.sa.}
\thanks{M. Majid Butt is with Nokia Bell Labs, France, and CONNECT Centre for Future Networks, Trinity College Dublin, Ireland. Email: Majid.Butt@@tcd.ie.}
\thanks{This publication has emanated from research conducted with the financial support of Science Foundation Ireland (SFI) and is co-funded under the European Regional Development Fund under Grant Numbers 13/RC/2077 and 14/US/I3110.}}

%

\IEEEtitleabstractindextext{%
\begin{abstract}
Caching at mobile devices and leveraging cooperative device-to-device (D2D) communications are two promising approaches to support massive content delivery over wireless networks while mitigating the effects of interference. To show the impact of cooperative communication on the performance of cache-enabled D2D networks, the notion of device clustering must be factored in to convey a realistic description of the network performance. In this regard, this paper develops a novel mathematical model, based on stochastic geometry and an optimization framework for cache-assisted coordinated multi-point (CoMP) transmissions with clustered devices. Devices are spatially distributed into disjoint clusters and are assumed to have a surplus memory to cache files from a known library, following a random probabilistic caching scheme. Desired contents that are not self-cached can be obtained via D2D CoMP transmissions from neighboring devices or, as a last resort, from the network. For this model, we analytically characterize the offloading gain and rate coverage probability as functions of the system parameters. An optimal caching strategy is then defined as the content placement scheme that maximizes the offloading gain. For a tractable optimization framework, we pursue two separate approaches to obtain a lower bound and a provably accurate approximation of the offloading gain, which allows us to obtain optimized caching strategies. Remarkably, if we replace the obtained expression for offloading gain with its lower bound, we can find a suboptimal caching strategy that is not only described via analytical formulas but can also show an improvement over the state-of-the-art caching schemes. Results reveal that cooperative transmission becomes more appealing in denser D2D caching networks and adverse interference conditions, which is the case of the imminent internet of things (IoT) and massive machine type communications era.
\end{abstract}

\begin{IEEEkeywords}
Coordinated multi-point (CoMP), probabilistic caching, offloading gain, clustered D2D communication.
\end{IEEEkeywords}}

\maketitle

\IEEEdisplaynontitleabstractindextext

\IEEEpeerreviewmaketitle

\IEEEraisesectionheading{
\vspace{-0.6 cm}
\section{Introduction}
\label{sec:introduction}}

\subsection{Background}
\IEEEPARstart{T}{he} proliferation of advanced mobile devices such as smartphones together with the popularity of video streaming causes tremendous growth of data traffic in cellular networks. \black{To address this challenge, the cellular industry is promoting the deployment of heterogeneous networks that are composed of different types of small \acp{BS}, such as micro-, pico-, and femto-\acp{BS} \cite{chandrasekhar2008femtocell}. Different types of small BSs are generally characterized by their transmission powers, coverage areas, and by whom they are deployed. For instance, the typical coverage range of a micro-BS is less than two kilometres. A pico-BS covers 200 meters or less whereas the range of a femto-BS is on the order of 10 meters. Micro- and pico- BSs can be deployed by the network operator while femto-BSs can be deployed by the end users.}

The deployment of heterogeneous networks results in a higher density of spatial reuse of radio resources and thus in higher overall network throughput \black{(i.e., the achievable rates)}. However, deploying a dense heterogeneous network comes with its own challenges. One such challenge is the deployment cost associated with connecting all the small cells to the backbone network with fast links, or the performance degradation accompanied by the capacity-limited backhauls \black{that connect these deployed BSs to the core of the network}. Besides, co-existence between small \acp{BS} and conventional macro base stations causes additional inter-cell interference when spectrum resources are shared. To address these challenges, caching popular content in advance at the network edge, e.g., at mobile devices or \acp{BS}, has been envisioned as a promising technique to relieve the backhaul congestion and improve user \ac{QoS} \cite{6871674} and \cite{amer20200caching}. \black{By caching, we mean storing the frequently demanded content at the network edge so as to remove the heavy burden and frequent requests on the (limited) backhauls. Broadly speaking, caching can be performed at the levels, ranging from, mobile devices to the level of small and macro \acp{BS}.} 

\black{In addition to caching, \ac{CoMP} transmission has been proposed to mitigate interference and increase network coverage and cell-edge throughput \cite{marsch2011coordinated}. Cooperative communication through \ac{CoMP} allows multiple transmission or reception of the same data from multiple BSs/devices in order to improve the spatial diversity. This essentially enables seamless co-existence of multiple in-band transmissions and allows efficient spectrum sharing and frequency reuse in such heterogeneous and D2D-enabled networks. In this paper, we are particularly interested in caching on mobile devices \cite{8412262} and \cite{Delay-Analysis}, together with cooperative transmission to boost the network traffic offloading and to overcome the performance degradation caused by co-channel interference \cite{nigam2014coordinated}.}  

Architecture of device caching exploits the large storage available in modern smartphones to cache multimedia files that might be highly demanded by end devices in the network \cite{8412262}. Devices can then exchange this multimedia content stored on their local storage with nearby devices. Since the distance between a requesting device and the device that stores a requested content is small in most cases, \ac{D2D} communication is commonly used for content transmission \cite{7342961}. As more than one device might cache the same content, the \ac{SIR} can be improved by joint transmission of the same cached content, which we refer to as  cooperative communication, e.g., via \ac{CoMP} transmission. For a detailed review of the main results and the literature on CoMP  transmission, the reader is referred to \cite{5490976} and \cite{nigam2014coordinated}.

Because of their respective advantages, both caching and cooperative transmission can be jointly adopted  in many practical scenarios. \black{For example, ensuring reliable delivery of ultra-high-definition streaming and \ac{VR} applications over wireless networks is very challenging due to the stringent \ac{QoS} requirements \cite{chaccour2019Reliability}. Leveraging D2D \ac{CoMP} transmissions of pre-downloaded frames from multiple devices to a requesting device, helps reduce communication delays and improve the perceived \ac{QoS}}. Proximity marketing, which is a wireless content advertising system associated with a particular place, might be another use case that leverages both caching and \ac{CoMP} transmissions \cite{hunter2013retail}. In particular, exploiting \ac{CoMP} transmissions to send pre-cached advertising content could lead to increasing the transmission range and mitigating interference among different operators of proximity marketing systems. Such cooperatively served contents are then delivered to individuals who wish to receive them, provided that they have the necessary equipment to do so \cite{hunter2013retail}. Motivated by the aforementioned discussions, it is important to study the role of cooperative transmission for cache-enabled D2D networks. The next section is devoted to summarizing relevant works in the literature that adopt cooperative transmission in wireless caching.


\vspace{-0.0 cm}
\subsection{Related Works}
The joint adoption of wireless caching and collaborative transmissions, where \acp{BS} (or devices) cooperatively serve a content, is widely adopted in literature \cite{chen2016cooperative,chae2017content,ao2015distributed,zheng2017optimization,daghal2018content,chen2017high,kim2018mds}. For instance, the authors in \cite{chen2016cooperative} proposed a combined caching scheme whereby part of the cache space is reserved for caching the most popular content, that is then cooperatively served from multiple \acp{BS}. Moreover, in \cite{chae2017content}, the authors investigated the trade-off between content diversity gain, i.e., serving diverse content, and cooperative gain, i.e., jointly transmitting the same content. In \cite{ao2015distributed}, the authors proposed cooperative transmissions for cache-enabled small cell networks to reduce the backhaul cost and delay. Meanwhile, the authors in \cite{zheng2017optimization} employed content caching at wireless relays to improve the overall performance of collaborative relayings for a network consisting of one source, one destination, and multiple relay nodes. 

Following a similar approach, employing cooperative content delivery in \ac{D2D} caching networks is discussed in \cite{daghal2018content,chen2017high,kim2018mds}. For instance, the authors in \cite{daghal2018content} proposed a multiple devices to a single device content delivery method via \ac{D2D} communication. Moreover, an opportunistic cooperation strategy for \ac{D2D} transmission is proposed in \cite{chen2017high} to mitigate interference among \ac{D2D} links. Combining coded caching along with \ac{CoMP} transmissions is recently studied in \cite{kim2018mds}, wherein redundantly stored data at caching helpers is utilized to combat wireless channel impairments due to channel fading and interference. 

While interesting, the works in \cite{daghal2018content,chen2017high,kim2018mds} did not consider the notion of device clustering, which is quite fundamental to the D2D network architecture \cite{zhang2015social} and \cite{hu2014evaluating}. \black{In this regard, the authors in \cite{clustered_twc} developed a stochastic geometry-based model to characterize the performance of content delivery in a clustered \ac{D2D} network whose devices are distributed according to a \ac{PCP}. Similarly, the authors in \cite{clustered_tcom} proposed a cluster-centric content placement scheme for \acp{PCP}, where the content of interest is cached closer to the cluster center. A clustered process such as \ac{PCP} essentially consists of two kinds of point processes, namely, the parent and daughter processes. If the underlying parent point process follows a standard \ac{PPP}, the whole process is called \ac{PCP} \cite{haenggi2012stochastic}. The works in \cite{clustered_twc} and \cite{clustered_tcom} assumed \acp{PCP} with equal number of devices per clusters. Meanwhile, the authors in \cite{8374852} proposed hybrid caching strategies to reduce the energy cost of D2D transmitters, where the location of these transmitters is modeled as a \ac{GPP}. For the \ac{GPP}, each cluster has \emph{only} one or two points, with probabilities $p$ and $1-p$, respectively. If a cluster consists of one point, that point is at the parent's location. If it has two points, one is at the parent's position (cluster center), and the other is uniformly distributed on a circle of certain radius centered at the parent point.}  
However, while the clustering nature of D2D communication is considered in the prior works \cite{clustered_twc} and \cite{clustered_tcom}, these works assumed that contents are pre-cached, i.e., there was no study of the caching problem. Moreover, modeling clustered D2D networks by means of \acp{GPP}, as done in \cite{8374852}, is limited by two facts: (i) The distance between transmitting and receiving devices within the same cluster is not captured by this model. (ii) The number of devices per cluster is assumed to be constant, particularly, fixed to only one (or two) device(s) per cluster. Furthermore, the content popularity and caching schemes in \cite{8374852}  were assumed to be the same for all clusters. However, in practice, users in different clusters might have different interests of files. For instance, users in a library might be interested in a different set of content from that of users in a pub. 

The offloading gain of a clustered D2D caching network modeled by \ac{TCP} is maximized in \cite{amer20202joint} by joint optimization of content caching and channel access. Moreover, the authors in \cite{amer2019optimizing} showed that the average service delay can be efficiently reduced by by jointly optimizing content caching and bandwidth partitioning. In \cite{8647532}, the authors proposed an efficient caching scheme for clustered D2D networks achieving minimum energy consumption. While these works studied the caching problem for clustered D2D networks, they only considered non-cooperative transmissions of requested contents. 




Compared with this prior art \cite{chen2016cooperative,chae2017content,ao2015distributed,zheng2017optimization,daghal2018content,chen2017high,kim2018mds,zhang2015social,hu2014evaluating,clustered_twc,clustered_tcom} and \cite{amer20202joint,amer2019optimizing,8647532,8374852}, this paper conducts performance analysis and statistical optimization of cache-assisted cooperative transmissions for a clustered D2D network. In particular, we characterize and optimize the offloading gain of a network of spatially clustered devices that adopt \ac{CoMP} transmissions and probabilistic caching. By maximizing the statistically-averaged offloading gain, our approach efficiently provides optimal averaged performance on a long-time scale to reduce signaling and processing overheads \cite{7158269}. Moreover, our model effectively captures the stochastic nature of channel fading and the clustered, yet random, network topology aspects, which have not been studied in the literature, particularly in the context of caching and \ac{CoMP} transmission. \emph{To the best of our knowledge, this paper provides the first rigorous analysis of  cache-assisted \ac{CoMP} transmissions for \ac{D2D} caching networks whose devices are modeled by a \ac{TCP}.}

\vspace{-0.2 cm}				
\subsection{Contributions}
The main contributions of this paper are summarized as follows:  
\begin{itemize}
\item We propose a cooperative transmission scheme via \ac{D2D} communications for clustered cache-enabled networks, whereby a device can be collaboratively served from multiple devices within the same cluster. We analytically characterize the offloading gain and rate coverage probability for the proposed network.

\item Given the complexity of the obtained rate coverage probability expression, we propose a tractable lower bound. We use this bound to prove that the interference power seen by the typical device of a \ac{TCP} can be upper-bounded by the interference power seen by the typical device of a \ac{PPP} with density that is the product of the \ac{TCP} cluster density and the average number of devices per cluster.

\item To further improve tractability and computational efficiency, we propose to approximate the signal power received from cooperative transmissions by two components: nearest and mean received power terms. Using Chebyshev's inequality, we prove that this approximation is remarkably tight and helps to reduce the original formulation to single integral. 

\item We use these closed-form expressions to define two suboptimal caching solutions for the offloading gain maximization problem. Ultimately, we show considerable improvements in the offloading gain under the optimized caching strategies compared with benchmark caching techniques. 
\end{itemize}

The rest of this paper is organized as follows. Section 2 and 3 present the system model and offloading gain characterization, respectively. The rate coverage analysis is conducted in Section 4, and the optimized caching probabilities are obtained in Section 5. Numerical results are then presented in Section 6, and conclusions are drawn in Section 7. 
\vspace{-0.4 cm}
\section{System Model}
\subsection{Network Model}
We consider a \ac{D2D} caching network in which devices are spatially distributed into disjoint clusters. The devices are assumed to have surplus memory that can be used to store content such as video files. Such a cached content is needed either for future use or to participate in content sharing with other devices within the same cluster. For this network, we model the location of the devices with a \ac{TCP} composed of parent and daughter points. An \ac{TCP} is generated by taking a parent homogeneous \ac{PPP} and daughter Gaussian PPP, one per parent, and translating the daughter processes to the position of their parent \cite{haenggi2012stochastic}. The cluster process is then the union of all the daughter points. \black{Importantly, the cluster centers are virtual points that determine the point around which cluster members are distributed. In other words, a cluster center resembles a geographical reference point for the spatial locations of the D2D devices within the same cluster. The cluster devices, however, are those cluster members (daughter points) that are scattered around the corresponding cluster centers.}

Let us denote the parent point process by $\Phi_p = \{\boldsymbol{x}_1,\boldsymbol{x}_2,\dots\}$, where $\boldsymbol{x}_i=\{x_1,x_2\}\in\R^2$, and $i\in\mathcal{N}$. Further, let $(\Phi_i)$ be a family of finite point sets representing the untranslated daughter Gaussian \acp{PPP}, denoted as $\Phi_c$, i.e., untranslated clusters. The cluster process is then the union of the translated clusters:
\begin{align}
\Phi \delequal  \cup_{i\in \mathcal{N}} \boldsymbol{x}_i + \Phi_i. 
\end{align}

The parent  and daughter points are referred, respectively, to as cluster centers and cluster members. We assume that the cluster centers are distributed according to a PPP $\Phi_p$ of density is $\lambda_p$. We also assume that, for Gaussian \acp{PPP}, the cluster members are normally scattered with variance $\sigma^2 \in \mathbb{R}$ around their cluster centers \cite{haenggi2012stochastic}. Given this normal scattering of daughter points, the \ac{PDF} of the cluster member location relative to its cluster center is given by
\begin{equation}
f_{\boldsymbol{Y}}(\boldsymbol{y}) = \frac{1}{2\pi\sigma^2}\textrm{exp}\Big(-\frac{\lVert \boldsymbol{y}\rVert^2}{2\sigma^2}\Big),	\quad \quad  \boldsymbol{y} \in \mathbb{R}^2,
\label{pcp}
\end{equation}
where $\boldsymbol{y}\in\R^2$ is the device location relative to its cluster center, $\lVert .\rVert$ is the Euclidean norm. Within each cluster, the number of cluster members is a Poisson \ac{RV} with a certain mean. For instance, if the average number of devices per cluster is $\bar{n}$, the cluster intensity will be:
\begin{align}
\lambda_c(y) = \frac{\bar{n}}{2\pi\sigma^2}\textrm{exp}\big(-\frac{\lVert \boldsymbol{y}\rVert^2}{2\sigma^2}\big),	\quad \quad  \boldsymbol{y} \in \mathbb{R}^2. 
\end{align}
Accordingly, the intensity of the entire process $\Phi$ will be $\lambda = \bar{n}\lambda_p$.

\vspace{-0.2 cm}
\subsection{Content Popularity and Probabilistic Caching}
Similarly to  \cite{8374852}, we assume two kind of devices co-exist within the same cluster, namely, \emph{content clients} and \emph{content providers}. In particular, the devices that can perform proactive caching and provide content delivery are called content providers while those requesting content, that also have caching capability, are called content clients. We assume that each device has a surplus memory of size $M$ files designated for caching content from a known file library $\mathcal{F}$. The total number of files is $N_f> M$ and the set of content indices is denoted as $\mathcal{F} = \{1, 2, \dots , N_f\}$. These $N_f$ files represent the content catalog that all the devices in a cell may request, which are indexed in a descending order of popularity. We assume that the probability that the $m$-th content is requested follows the standard Zipf distribution as widely-adopted in the literature, which is given by \cite{breslau1999web} 
\begin{equation}
\label{zipf}
q_m = \Bigg(m^{\beta} \sum_{k=1}^{N_f}k^{-\beta}\Bigg)^{-1},
\end{equation}
where $\beta$ is a parameter reflecting how skewed the popularity distribution is. The larger $\beta$, the fewer files that are responsible for the majority of requests \cite{breslau1999web}; by definition, $\sum_{m=1}^{N_f}q_m$.  Moreover, we assume that the content popularity may vary across clusters. For instance, users in a library may be interested in an entirely different set of files from the users in a sports center. Therefore, the Zipf distribution models the per-cluster popularity of files. Such a discrepancy of contents of interest can be captured by having different popularity indexes $\beta$ per different clusters, i.e., different concentration rates. Ranks of popular contents can be also different among different clusters. This discrepancy of popular files implies that the content request and, consequently, caching design models vary across clusters. Such a cluster-specific popularity can be seen as a direct generalization of the individual user preferences that is modeled in \cite{8667721}. 	

The cluster-specific popularity model necessitates the design of content placement on a per-cluster basis. Hence, within each cluster, we assume a random content placement scheme in which file $m$ is cached independently at each cluster device according to the probability $c_m$, with $0 \leq c_m \leq 1, \forall m \in \{1, \dots, N_f\}$. To avoid duplicate caching of the same file within the memory of the same device, we follow the \ac{PC} approach proposed in \cite{geographic_caching}, which requires that $\sum_{m=1}^{N_f}c_m=M$. It is worth mentioning that the \ac{PC} is a standard caching technique that is widely adopted in the literature \cite{chen2017probabilistic,blaszczyszyn2015optimal,chae2016caching}.



\subsection{Content Request and Delivery Model} 
Enabling seamless video delivery over cellular networks implies stringent \ac{QoS} requirements. However, the performance of wireless networks, especially D2D communications, is limited by interference and the effects of small scale fading. Therefore, cooperative communication turns to be more appealing as a prominent interference mitigation tool. We hence allow multiple devices to jointly serve their cached content to a common device within the same cluster via non-coherent \ac{CoMP} transmission. The underlying reason of assuming a non-coherent transmission is that it is hard to estimate the \ac{CSI} for the \ac{D2D} communications. We consider out-of-band \ac{D2D} communication system, i.e., there is no cross-interference between the cellular network and \ac{D2D} communication. All devices are equipped with a single transmit-receive isotropic antenna, and they have no \ac{CSI} from the device they are serving. Furthermore, each D2D transmission uses all the available bandwidth. Transmitted signals experience single-slope path loss with attenuation exponent $\alpha > 2$ and small scale fading, which we model as an \ac{i.i.d.} complex Gaussian \ac{RV} with zero mean and unit variance. \black{In our study, we are particularly focused on developing a general model that can be then tailored for any physical layer transmission scheme. For instance, if the mobile devices use perfectly-synchronized \ac{OFDM} symbols with sufficiently large \ac{CP}, the \ac{ISI} can be effectively cancelled out. Moreover, since we assume static devices, there would be neither doppler effect nor inter-carrier interference (ICI) when adopting \ac{OFDM} and proper channel selections.}

Due to the cost of participating in content caching and delivery, e.g., battery consumption and memory utilization, not all content providers  can be active in all time slots. Hence, within each cluster, we assume that  content providers can be available for content delivery with probability $p\in[0,1]$. In other words, among $\bar{n}$ average number of devices per cluster, only average of $p\bar{n}$ devices are willing to participate in content delivery and caching. Further, we assume a \ac{BS}-assisted \ac{D2D} link setup scheme, where the transmissions of different files in different clusters are orchestrated by the \ac{BS} \cite{8458381}.  %
In details, a \emph{content client} first sends its request to its geographically closest \ac{BS}, which knows the active content providers within the same cluster, their cached files, as well as their locations. If there are active content providers caching the requested file, the BS then establishes direct CoMP D2D links between the content client and the set of active content providers for this requested file. Requests for contents are assumed to be of negligible size, so that they do not add signalling overhead and are always successfully decoded at the BS. 

Within each cluster, we assume a content client device whose distance to its cluster center is drawn from a Rayleigh distribution of scale parameter $\sigma$, according to the \ac{TCP} definition \cite{haenggi2012stochastic}. Throughout time, content clients in different clusters may request files $i\in\mathcal{F}$ with a probability following the assumed \emph{per-cluster Zipf distribution} \black{in Eq. (\ref{zipf})}. Since each cluster has its own library, a given content client may either be served via D2D connection(s) from active provider(s) within the same cluster or, as a last resort, via the nearest geographical BS. To recap, under the proposed  transmission and caching schemes, one content client per cluster is cooperatively served at a time from neighboring active content providers while being interfered only from active providers in other (remote) clusters (i.e., inter-cluster interference).

Notice that, according to the independent thinning theorem \cite[Theorem 2.36]{haenggi2012stochastic}, active devices within the same cluster form a Gaussian \ac{PPP}  $\Phi_{cp}$ whose intensity function is $\lambda_{cp}(\boldsymbol{y})=p \lambda_c(\boldsymbol{y})$. Similarly, the set of active providers that cache a desired content $m$ are modeled as a Gaussian \ac{PPP}  $\Phi_{cpm}$ with the intensity function given by $\lambda_{cpm}(\boldsymbol{y})=c_m p \lambda_c(\boldsymbol{y})$. Hence, within each cluster, the number of active devices and the number of active devices caching content $m$ are Poisson \acp{RV} of means $p\bar{n}$ and $c_mp\bar{n}$, respectively. By definition, $\Phi_{cpm} \subseteq \Phi_{cp} \subseteq \Phi_{c}$.


The main advantages of D2D caching networks lie in alleviating the burden of the backhaul links  and improving the network spectral efficiency. To leverage these features, it is crucial to intelligently cache and deliver contents to maximize the percentage of offloaded traffic from the network core to the edge. The offloading gain is widely-adopted as a key performance metric to quantify this percentage \cite{chen2016cooperative,chae2017content}, and \cite{8374852}. Specifically, the offloading gain is defined as the probability of obtaining a desired content either from the self-cache or via \ac{D2D} communication with a received \ac{SIR} greater than a target threshold. Hence, our target in the next section is to characterize and maximize the offloading gain of the proposed CoMP-assisted D2D caching network. 

\vspace{-0.0 cm}
\section{Offloading Gain Characterization}
\begin{figure} [!t] 
\vspace{-0.0 cm}
\centering
\includegraphics[width=0.35\textwidth]{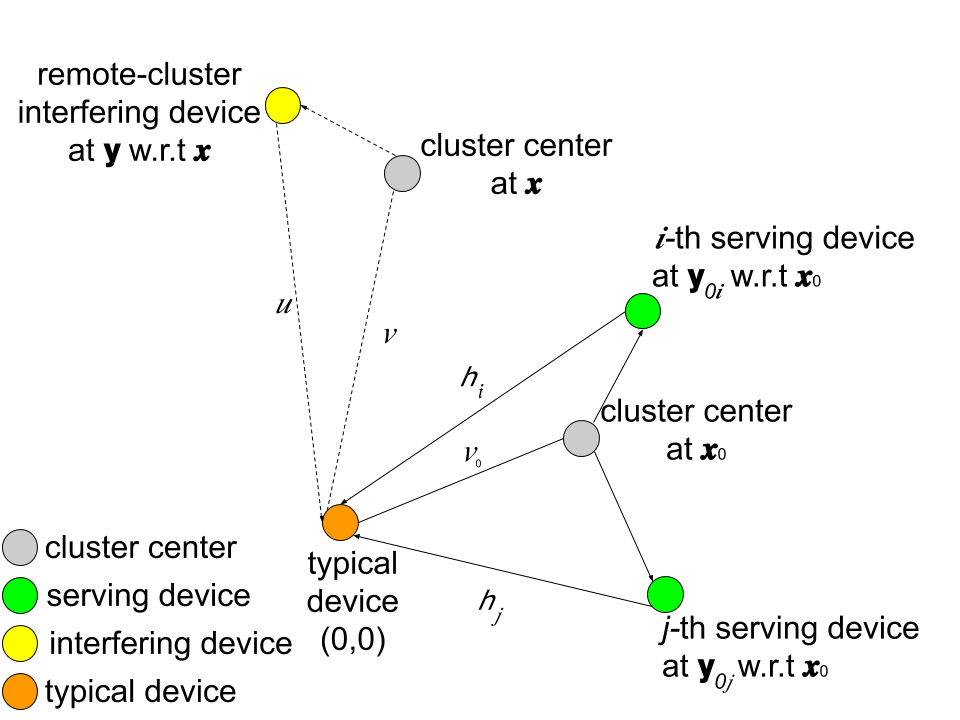}		
\caption {Illustration of the representative cluster and one interfering cluster, where $\{\boldsymbol{x}_0,\boldsymbol{y}_{0i}, \boldsymbol{y}_{0j}, \boldsymbol{x},\boldsymbol{y}\} \in \mathbb{R}^2$ and $\{v_0,h_i,h_j,v,u\} \in \mathbb{R}$.}
\label{distance_near}
\vspace{-0.3 cm}
\end{figure}

%
%
%
%
%

Given stationarity of the parent process and independence of the daughter process, we can conduct our analysis for the representative cluster, which is an arbitrary cluster whose center is located at $\boldsymbol{x}_0\in \Phi_p$, and a \emph{typical client}, which is a randomly selected member of the representative cluster that requests the content. \black{We are particularly interested in the network averaged offloading gain by focusing on the performance of this typical device. The study of the network overall capacity and size is also an interesting problem but is beyond the scope of this paper.}  Without loss of generality, we assume the typical client is located at the origin $(0, 0) \in \mathbb{R}^2$. 

When some active content providers jointly transmit a desired content $m$, the signal received at the typical client consists of two main components: the desired signal that represents the joint non-coherent transmissions from   active providers that cache content $m$ in the representative (local) cluster, and the interference component that is created by other active providers in remote clusters.  This can be formally stated as:
\begin{align}
&y_m =\underbrace{\sum_{\boldsymbol{y}_{0i} \in \Phi_{cpm}}\sqrt{\gamma_d} G_{\boldsymbol{y}_{0i}} \lVert \boldsymbol{x}_0 + \boldsymbol{y}_{0i}\rVert^{-\alpha/2}  s_{\boldsymbol{y}_{0i}}}_{\text{desired signal}} \nonumber \\
&
 \quad\quad\quad  \quad\quad \quad\quad + \underbrace{\sum_{\boldsymbol{x} \in \Phi_p^{!}} \sum_{\boldsymbol{y} \in \Phi_{cp}} \sqrt{\gamma_d} G_{\boldsymbol{y}} \lVert\boldsymbol{x}+\boldsymbol{y}\rVert^{-\alpha/2} s_{\boldsymbol{y}}}_{\text{inter-cluster interference}} + z, 
 \nonumber 
\end{align}				 
where $i \in \{1,\dots,|\Phi_{cpm}|\}$, $G_{\boldsymbol{y}_{0i}}$ denotes the power fading between an active provider at $\boldsymbol{y}_{0i} \in \Phi_{cpm}$ relative to its cluster center at $\boldsymbol{x}_0$ and the typical client, see Fig.~\ref{distance_near}; $\gamma_d$ denotes the \ac{D2D} transmission power, and $s_{\boldsymbol{y}_{0i}}$ is the symbol jointly transmitted by the active providers $\boldsymbol{y}_{0i} \in \Phi_{cpm}$. $\Phi_p^{!}= \Phi_p \setminus\{\boldsymbol{x}_0\}$ denotes the set of remote clusters, and $\Phi_{cp} \subseteq \Phi_{c}$ represents the set of active devices in a remote cluster centered at $\boldsymbol{x}\in\Phi_p^{!}$. Finally, $z$ denotes the standard additive white Gaussian noise.

Note that representing the set of inter-cluster interferers as $\Phi_{cp}$ corresponds to the worst case interference scenario, when all active devices in a remote cluster are caching the required content of
their own-cluster content client. This bound is in line with our analysis and the underlying network model, particularly,  the assumption of different content popularity and placement per clusters. This is because, based on this bound, the inter-cluster interference power, and correspondingly, the per cluster cache design will be independent of the content demand in other clusters. 

\black{We focus on the interference-limited regime, thus omitting the thermal noise. This assumption of negligible thermal noise power compared to the interference power is widely-adopted in the literature, especially for clustered D2D networks \cite{clustered_twc,clustered_tcom}, and \cite{8374852}. Moreover, accounting for the noise power is a direct extension of the analysis in this paper}. Assuming unit power Gaussian symbols, the received \ac{SIR} at the typical client when downloading content $m$ is given by
\begin{align} 	
\sir_m = \frac{\gamma_d \Bigg|\sum_{\boldsymbol{y}_{0i} \in \Phi_{cpm}}  G_{\boldsymbol{y}_{0i}} \lVert \boldsymbol{x}_0 + \boldsymbol{y}_{0i}\rVert^{-\alpha/2}\Bigg|^2}{I_{\rm out}},
\end{align}
where \black{$I_{\rm out}$} is the sum of interfering signal power associated with the downloading of content $m$, given by: 
\begin{align}
  I_{\rm out} =&  
  \gamma_d \Big| \sum_{\boldsymbol{x} \in \Phi_p^{!}} \sum_{\boldsymbol{y} \in \Phi_{cp}} G_{\boldsymbol{y}} \lVert\boldsymbol{x}+\boldsymbol{y}\rVert^{-\alpha/2} \Big|^2. 
 \nonumber
\end{align}

Finally, the offloading gain can be formally stated as:
\begin{align}
\label{eq:offload}
\mathbb{P}_{o}(\boldsymbol{c}) =  \sum_{m=1}^{N_f} q_m c_m + q_m(1-c_m)\Upsilon_m,
 \end{align}
where $\boldsymbol{c}=\{c_1,\dots,c_m,\dots,c_{N_f}\}$, and $\Upsilon_m=\mathbb{P}(\sir_m>\vartheta)$ is the rate coverage probability for content $m$, i.e., the probability that the received \ac{SIR} via \ac{CoMP} transmissions is larger than a target threshold $\vartheta$, which we characterize in the sequel. \black{In Eq. (\ref{eq:offload})}, the first term corresponds to the event of serving  a desired content from local memory, i.e., self-cache \cite{8422592}. The second term   represents the joint event that the desired content is not locally cached while being cached and downloadable from active providers in the same cluster, with an \ac{SIR} greater than the target threshold $\vartheta$. 

\vspace{-0.3 cm}
 \section{Rate Coverage Probability Analysis}
Our objective in this section is to analytically characterize the offloading gain. In particular, we first derive the exact expression of $\mathbb{P}_{o}(\boldsymbol{c})$ as a function of the system parameters. Then, we seek lower bound  and approximation of the rate coverage probability $\Upsilon_m$ that will result in easy-to-compute expressions of the offloading gain, and provide useful system design insights.

In the case of \ac{CoMP} transmissions, active providers in the representative cluster jointly transmit the requested content to the typical client. The received power at the typical client is the sum of the received signal powers from active providers, and hence, the rate coverage probability $\Upsilon_m$ is: 
 \begin{align}
 \label{rate-cov-prob}
\Upsilon_m &= \Pb\Bigg[ \frac{ \gamma_d \Big|\sum_{\boldsymbol{y}_{0i} \in \Phi_{cpm}}  G_{\boldsymbol{y}_{0i}} \lVert \boldsymbol{x}_0 + \boldsymbol{y}_{0i}\rVert^{-\alpha/2}\Big|^2}{I_{\rm out}} \geq \vartheta\Bigg].
 \end{align}
Since $G_{\boldsymbol{y}_{0i}}$ are \ac{i.i.d.} complex Gaussian \acp{RV}, we get
\begin{align}
\label{ch-gain}
\Bigg|\sum_{\boldsymbol{y}_{0i} \in \Phi_{cpm}}  \lVert \boldsymbol{x}_0 + \boldsymbol{y}_{0i}\rVert^{-\alpha/2} & G_{\boldsymbol{y}_{0i}}\Bigg|^2\sim
\nonumber \\
 & \exp\Bigg(\frac{1}{\sum_{\boldsymbol{y}_{0i} \in \Phi_{cpm}}  \lVert \boldsymbol{x}_0 + \boldsymbol{y}_{0i}\rVert^{-\alpha}}\Bigg).
\end{align}
Hence, from (\ref{rate-cov-prob}) and  (\ref{ch-gain}), we have     
\begin{align}
\label{eq:rp-exact}
 \Upsilon_m & \overset{}{=}  \Eb\Big[\exp\Big(-\frac{\vartheta \big( I_{\rm out}\big)}{\gamma_d S_{\Phi_{cpm}}}\Big)\Big] \nonumber \\
&\overset{(a)}{=}  \Eb\left[ 			
\Lc_{I_{\rm out}} (t)			
\Bigg|S_{\Phi_{cpm}}=s_{\Phi_{cpm}}\right],
\end{align} 
where $S_{\Phi_{cpm}}=\sum_{\boldsymbol{y}_{0i} \in \Phi_{cpm}}  \lVert \boldsymbol{x}_0 + \boldsymbol{y}_{0i}\rVert^{-\alpha}$ is a \ac{RV} that can be physically interpreted as the received signal power from the active providers  devices $\boldsymbol{y}_{0i} \in \Phi_{cpm}$ subject to path loss only (as we have already averaged over the fading based on the \ac{PDF} \black{in Eq. (\ref{ch-gain})}), assuming normalized power. (a) follows from the Laplace transform of the  interference $I_{\rm out}$ evaluated at $t=\frac{\vartheta}{\gamma_d s_{\Phi_{cpm}}}$.  
We derive the Laplace transform of interference in the following Lemma to compute the rate coverage probability, and correspondingly, the offloading gain. 

\begin{lemma}
\label{ch4:comp-interference}				
Laplace transform of the inter-cluster interference, conditioned on a realization of the active providers for content $m$ in the representative cluster, is given by
\begin{align}
\label{eq:lap-transform}
\Lc_{I_{\rm out}}(t) = {\rm exp}\Big(-2\pi \lambda_p\int_{v=0}^{\infty}\Big(1 - e^{-p\bar{n}\zeta(v,t)}\Big)v\dd{v}\Big),
\end{align}
where $t=\frac{\vartheta}{\gamma_d s_{\Phi_{cpm}}}$, $\zeta(v,t) = \int_{u=0}^{\infty}\frac{t\gamma_d}{u^{\alpha}+t\gamma_d} f_{U|V}(u|v)\dd{u}$, $f_{U|V}(u|v)=\mathrm{Rice} (u;v,\sigma)$ is the Rician \ac{PDF} modeling the distance $U=\lVert\boldsymbol{x}+\boldsymbol{y}\rVert$ between an interfering device at $\boldsymbol{y}$ relative to its cluster center at $\boldsymbol{x} \in \Phi_{p}$ and the origin $(0,0)$, conditioned on $V=\lVert\boldsymbol{x}\rVert=v$.
\end{lemma}
\begin{IEEEproof}
Please refer to \App{proof-comp-interference}. 
\end{IEEEproof}

We continue by characterizing the joint serving distance distribution. For a given realization $S_{\Phi_{cpm}}=s_{\Phi_{cpm}}$, let us assume that there are $k$ active providers in the representative cluster. 
Let us also denote joint distances from the typical client (origin) to the $k$ content providers in the representative cluster, centered at $\boldsymbol{x}_0$, as $\boldsymbol{H}_k= \{H_1, \dots, H_k\}$. Then, conditioning on $\boldsymbol{H}_k = \boldsymbol{h}_k$, where $\boldsymbol{h}_k= \{h_1, \dots, h_k\}$, the conditional (i.e., on $k$) \ac{PDF} of the joint serving distance is denoted as $f_{\boldsymbol{H}_k}(\boldsymbol{h}_k)$. 
Hence, conditioning on $k$, we can express the rate coverage probability as 
\begin{align}
\Upsilon_{m|k} =  \Eb\left[ \Lc_{I_{\rm out}} \Big(t=\frac{\vartheta}{\gamma_d\sum_{i=1}^{k}  h_i^{-\alpha}}\Big)\Bigg|S_{\Phi_{cpm}}=s_{\Phi_{cpm}}\right],
\end{align} 
where $s_{\Phi_{cpm}}=\sum_{i=1}^{k}  h_i^{-\alpha}$, $ h_i=\lVert \boldsymbol{x}_0 + \boldsymbol{y}_{0i}\rVert$, and $\boldsymbol{y}_{0i}\in\Phi_{cpm}$. Since a content provider $i$ in the representative cluster centered at $\boldsymbol{x}_0$ has its coordinates in $\mathbb{R}^2$ chosen independently from a Gaussian distribution with standard deviation $\sigma$, then, by definition, the distance from such a content provider to the origin, denoted as $h_i=\lVert \boldsymbol{x}_0+\boldsymbol{y}_{0i}\rVert$, has Rician distribution $f_{H_i|V_0}(h_i|v_0)=\mathrm{Rice}(h_i;v_0,\sigma)$. 
Since also the content providers and the typical client have their locations sampled from a normal distribution with variance $\sigma^2$ relative to  their cluster center $\boldsymbol{x}_0$, then, by definition, the statistical distance distribution between any two points, e.g., from the $i$-th content provider to the typical client, follows Rayleigh distribution with scale parameter $\sqrt{2}\sigma$, written as 
\begin{align}
f_{H_i}(h_i) = \mathrm{Rayleigh}(h_i,\sqrt{2}\sigma) = \frac{h_i}{2\sigma^2}e^{-\frac{h_i^2}{4\sigma^2}}.
\end{align}
If the serving distances from the typical client to the different points of the cluster were independent from each other, $f_{\boldsymbol{H}_k}(\boldsymbol{h}_k)$ would simply be the product of $k$ independent \acp{PDF}, each of which having $f_{H_i}(h_i)=\mathrm{Rayleigh}(h_i,\sqrt{2}\sigma)$. 
However, there is a correlation between the serving distances due to the common factor $\boldsymbol{x}_0$ in the serving distance equation $h_i=\lVert \boldsymbol{x}_0 + \boldsymbol{y}_{0i}\rVert$ with $\boldsymbol{y}_{0i}\in \Phi_{cpm}$, see also Fig.~\ref{distance_near}. For analytical tractability, similar to \cite{clustered_twc} and \cite{amer2019optimizing}, we neglect this correlation. Hence, the conditional \ac{PDF} of the joint serving distance $f_{\boldsymbol{H}_k}(\boldsymbol{h}_k)$ can be  obtained from
\begin{align}
\label{joint-pdf} 
f_{\boldsymbol{H}_k}(\boldsymbol{h}_k) = \prod_{i=1}^{k}\frac{h_i}{2\sigma^2}e^{-\frac{h_i^2}{4\sigma^2}}.
\end{align}
\black{Such a negligible correlation assumption yields an upper bound on the exact rate coverage probability of the typical device. Moreover, if the typical device is assumed to be at the cluster center, there would be no correlation between the underlying serving distances, e.g., see \cite{amer2020caching} and \cite{8998329}, and (13) will represent the exact joint serving distance \ac{PDF}}. Conditioning on having $k$ active content providers, i.e., $s_{\Phi_{cpm}}=\sum_{i=1}^{k}  h_i^{-\alpha}$, the rate coverage probability will be given by $\Upsilon_{m|k} = $
\begin{align}
\label{given-k} 
\int_{\boldsymbol{h}_k=\boldsymbol{0}}^{\infty} \mathscr{L}_{I_{\rm out}} \Bigg(\frac{\vartheta}{\gamma_d\sum_{i=1}^{k}  h_i^{-\alpha}}
\Bigg |S_{\Phi_{cpm}}=s_{\Phi_{cpm}} \Bigg) f_{\boldsymbol{H}_k}(\boldsymbol{h}_k) \dd{\boldsymbol{h}}_k.
\end{align}

Given that $\Phi_{cpm}$ is a Gaussian \ac{PPP} , the number of active content providers for content $m$ is a Poisson \ac{RV} with mean $c_mp\bar{n}$. Therefore, the probability that there are $k$ content providers is equal to $\frac{(p\bar{n}c_m)^ke^{-p\bar{n}c_m}}{k!}$. Invoking this along with (\ref{eq:lap-transform}), (\ref{joint-pdf}), and (\ref{given-k}) into (\ref{eq:offload}), $\mathbb{P}_{o}(\boldsymbol{c})$ is given as 
\begin{align}
&\mathbb{P}_{o}(\boldsymbol{c}) =  \sum_{m=1}^{N_f} q_m\Bigg(c_m + \big(1-c_m\big).\sum_{k=1}^{\infty} \frac{(p\bar{n}c_m)^ke^{-c_mp\bar{n}}}{k!} \cdot
& \nonumber		\\
\label{ch4:offloading-gain-eqn} 
 & \int_{\boldsymbol{h}_k=\boldsymbol{0}}^{\boldsymbol{\infty}} e^{-2\pi \lambda_p\int_{v=0}^{\infty}\Big(1 - e^{-p\bar{n}(1 - \zeta(v,t))}\Big)v\dd{v}} \prod_{i=1}^{k}\frac{h_i}{2\sigma^2}e^{-\frac{h_i^2}{4\sigma^2}} \dd{\boldsymbol{h}}_k \Bigg).
\end{align} 

Since the obtained expression \black{in Eq. (\ref{ch4:offloading-gain-eqn})} involves multi-fold integrals and summations, this renders the calculation of the rate coverage probability computationally complex. Furthermore, the offloading gain maximization problem turns to be intractable. Therefore, in the sequel, we focus on tight bound and approximation of the rate coverage probability that will result in easy-to-compute expressions that also enable us to   formulate a tractable optimization problem to maximize the offloading gain. 

\vspace{-0.4 cm}
\subsection{Lower Bound on Offloading Gain}		
Next, we obtain a tractable lower bound on the offloading gain based on an upper bound on the interference power. \black{By replacing the exponential interference formula in the Laplace transform expression  of inter-cluster interference in (\ref{eq:lap-transform}) by the first and second terms of its Taylor's series expansion,  we obtain an upper bound on the interference power. Accordingly, this bound on the Laplace transform of inter-cluster interference yields a lower bound on the rate coverage probability $\Upsilon_i$.}  
\begin{theorem}
\label{ch4:comp-interference-approx}	
Laplace transform of interference derived in \Eq{lap-transform} can be bounded by
 \begin{equation} 
 \label{eq:ppp-interference}
\Lc_{I_{\rm out}}(t) \approx \exp\left(-\pi p\bar{n}\lambda_p t^{2/\alpha} \Gamma(1 + 2/\alpha)\Gamma(1 - 2/\alpha)\right),
\end{equation}
and, correspondingly, a lower bound on the offloading gain is given by
\begin{align}
& \mathbb{P}_{o}^{\sim}(\boldsymbol{c}) = \sum_{m=1}^{N_f} q_m\Big(c_m + \big(1-c_m\big).\sum_{k=1}^{\infty} \frac{(c_mp\bar{n})^ke^{-c_mp\bar{n}}}{k!} \times 
\nonumber \\
\label{ch4:lower-bound-offload-gain} 
& \int_{\boldsymbol{0}}^{\infty} e^{-\pi p \bar{n}\lambda_p (\frac{\vartheta}{\sum_{i=1}^{k} h_i^{-\alpha}})^{2/\alpha} \Gamma(1 + 2/\alpha)\Gamma(1 - 2/\alpha)}
 \prod_{i=1}^{k}\frac{h_i}{2\sigma^2}e^{-\frac{h_i^2}{4\sigma^2}} \dd{\boldsymbol{h}}_k\Big).
 \end{align}
\end{theorem}
\begin{IEEEproof}
Please refer to \App{proof-comp-interference-approx}. 
\end{IEEEproof}

\begin{remark} {\rm The obtained expression \black{in Eq. (\ref{eq:ppp-interference})} boils down to the Laplace transform of a \ac{PPP} with intensity $\bar{n}\lambda_p$. This shows that the inter-cluster interference of a \ac{TCP} with density of clusters $\lambda_p$ and average number of devices per cluster $\bar{n}$, i.e., with intensity $\bar{n} \lambda_p$, is upper bounded by that of a \ac{PPP} of the same intensity.}
\end{remark}

\begin{figure}[t]	
\vspace{-0.3 cm}
\centering
\includegraphics[width=0.45\textwidth]{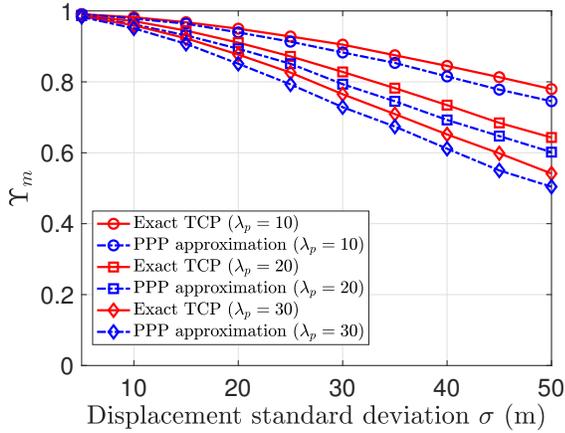}   
\caption {The lower bound on $\Upsilon_m$ based on (\ref{eq:ppp-interference}) versus displacement standard deviation $\sigma$ is plotted for various parent \ac{PPP} densities $\lambda_p$ ($\bar{n}=20$, $p=0.5$, $c_m=0.5$). "Exact TCP" in the legend refers to the exact performance for the \ac{TCP} while "PPP approximation" refers to the lower bound based on Theorem \ref{ch4:comp-interference-approx}.}			
\label{fig:lower-bound}
\vspace{-0.4 cm}
\end{figure}
In Fig.~\ref{fig:lower-bound}, we plot the exact expression and its lower bound, based on (\ref{eq:ppp-interference}), of the rate coverage probability  versus displacement variance $\sigma$ for various parent \ac{PPP} $\Phi_p$ densities $\lambda_p$. The derived lower bound is considerably tight when both $\sigma$ and $\lambda_p$ are relatively small. Also, it is noticeable that $\Upsilon_m$ monotonically decreases with both $\sigma$ and $\lambda_p$, which reflects the fact that the desired signal is weaker when the distance between content providers and the typical client is larger, and the effect of inter-cluster interference increases when the density of clusters increases, respectively. When $\lambda_p$ and $\sigma$ increase, the obtained lower bound becomes no longer tight, however, it still represents a reasonable bound on the exact $\Upsilon_m$.

Having obtained a lower bound on the offloading gain, next, we seek further approximation by replacing the desired signal power with the sum of the signal power from the nearest device signal and the mean power received from all other active content providing devices.
\vspace{-0.3 cm}
\subsection{Serving Power Approximation}			
From (\ref{ch-gain}), $S_{\Phi_{cpm}}=\sum_{\boldsymbol{y}_{0i} \in \Phi_{cpm}}  \lVert \boldsymbol{x}_0 + \boldsymbol{y}_{0i}\rVert^{-\alpha}$ represents a \ac{RV} that models the intended signal power from content providers $\boldsymbol{y}_{0i} \in \Phi_{cpm}$ \emph{subject to path loss only}. Next, we adopt the so-called \emph{mean plus nearest approximation} to approximate this intended signal power $S_{\Phi_{cpm}}$ as a sum of two terms. \black{More specifically, the first term $\lVert \boldsymbol{x}_0 + \boldsymbol{y}_{01}\rVert^{-\alpha}$ is the received signal power from the nearest active provider, where  $\boldsymbol{y}_{01}=\argmin_{\boldsymbol{y}_{0i} \in \Phi_{cpm}}\{\lVert \boldsymbol{x}_0 + \boldsymbol{y}_{0i}\rVert\}$. In addition, the second term is the statistical average of the \ac{RV} $S_{\Phi_{cpm}^{!}} $, where $S_{\Phi_{cpm}^{!}}  = \sum_{\boldsymbol{y}_{0i} \in \Phi_{cpm}\setminus\boldsymbol{y}_{01}}   \lVert \boldsymbol{x}_0 + \boldsymbol{y}_{0i}\rVert^{-\alpha}$ is the overall signal power received from all other active providers conditioning on the nearest serving  distance $H_1=h_1=\lVert \boldsymbol{x}_0 + \boldsymbol{y}_{01}\rVert$.  
The main motivation behind this approximation is that if the \ac{RV} $S_{\Phi_{cpm}^{!}} $ is concentrated around its mean, it can be effectively approximated by this mean $\Eb[S_{\Phi_{cpm}^{!}} |h_1]$. The tightness of this approximation is investigated in the sequel.} 

As we will see, this approximation yields an easy way to obtain the rate coverage probability while also being tight. This approach has been similarly adopted to circumvent intractable analysis in the stochastic geometry literature, see, e.g., \cite{8536464}. Starting from the Laplace transform expression \black{in Eq. (\ref{eq:rp-exact})}, we approximate $S_{\Phi_{cpm}}$ as		
\begin{equation}
\label{s_phi_approx}
S_{\Phi_{cpm}} \approx \Big\lVert \boldsymbol{x}_0 + \boldsymbol{y}_{01}\Big\rVert^{-\alpha} + \Eb\Big[ S_{\Phi_{cpm}^{!}}|  \boldsymbol{y}_{0i} \Big],  
 \end{equation}					
where $\Phi_{cpm}^{!} = \Phi_{cpm} \setminus \boldsymbol{y}_{01}$, and $S_{\Phi_{cpm}^{!}}= 
\sum_{\boldsymbol{y}_{0i} \in \Phi^!_{cpm}}\lVert \boldsymbol{x}_0 + \boldsymbol{y}_{0i}\rVert^{-\alpha}$.
Next, we derive the distribution of nearest serving  distance $h_1=\lVert \boldsymbol{x}_0 + \boldsymbol{y}_{01}\rVert$. Then, we prove the concentration of the proposed approximation using Chebyshev's inequality. 
Finally, given the distance distribution to the nearest device $f_{H_1}(h_1)$, and the derived formula for $\Eb\left[ S_{\Phi_{cpm}^{!}}\Big| H_1=h_1 \right]$, an approximation for $S_{\Phi_{cpm}}$ is obtained based on  (\ref{s_phi_approx}).
\begin{lemma}
\label{ch4:pdf-nearest-sitance}				
The \ac{PDF} of the distance from the typical client to the nearest active provider in $\Phi_{cm}$ is given by 
\begin{align}
\label{Leibniz}
f_{H_1}(h_1) &=c_mp\bar{n}\int_{v_0=0}^{\infty} f_{V_0}(v_0) f_{H_1|V_0}(h_1|v_0)  \times 
\nonumber \\
&   \quad\quad\quad\quad\quad\quad\quad e^{-c_mp\bar{n}\int_{0}^{h_1}f_{H|V_0}(h|v_0)\dd{h}}\dd{v_0},
 \end{align}
 \black{and using Jensen's inequality, it can be approximated by}
 \begin{equation}
 \label{nearest-1}
f_{H_1}(h_1) \approx \frac{c_mp\bar{n} h_1 \exp\left(-c_mp\bar{n}\left(1-\exp(\frac{-h_1^2}{4 \sigma^2})\right) -\frac{h_1^2}{4 \sigma^2} \right)}{2 \sigma^2}.
 \end{equation}
\end{lemma}
\begin{IEEEproof}
The proof is provided in \App{proof-pdf-nearest-distance}. 
\end{IEEEproof}

\begin{figure}[!tb]	
\vspace{-0.3 cm}
\centering
\includegraphics[width=0.45\textwidth]{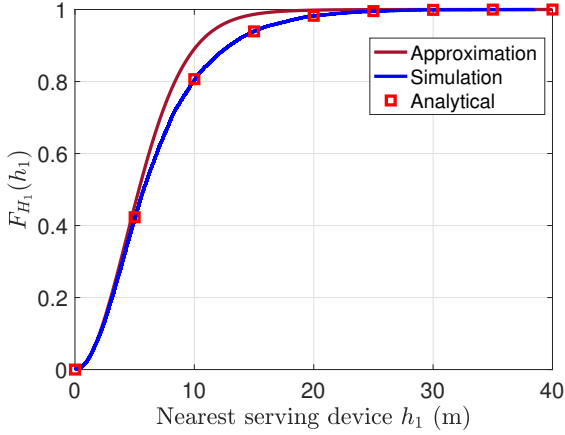}		
\caption {The derived nearest serving distance \ac{CDF} \black{in Eq. (\ref{near-cdf})} is plotted and compared with simulation and Jensen's inequality-based approximation \black{in Eq. (\ref{approx_CDF})} ($\bar{n}=20$, $\sigma=\SI{10}{m}$, $c_m=0.5$, $p=1$).}		
\label{fig:near_dist_cdf}
\vspace{-0.5 cm}
\end{figure}

The accuracy of the derived \ac{CDF} $F_{H_1}(h_1)$ \black{in Eq. (\ref{near-cdf})} and its approximation based on Jensen's inequality \black{in Eq. (\ref{approx_CDF})} (see \App{proof-pdf-nearest-distance}) are verified in Fig.~\ref{fig:near_dist_cdf}. It is clear from (\ref{nearest-1}) that the distance to the nearest active content provider statistically  decreases as $c_m$ or $p$ increase, i.e., when there is a high probability of having active and caching content providers within the local cluster. The distance is also more likely to decrease as $\bar{n}$ increases since a congested cluster has shorter distance between the content client and providers. 
 
Next, we show that approximating the desired signal by its nearest and conditional mean components, see (\ref{s_phi_approx}), yields an accurate yet tractable expressions for the rate coverage probability and offloading gain.
 
 \begin{proposition}
\label{ch4:concentration}
For scenarios of practical interest, the proposed approximation for $S_{\Phi_{cpm}}$ \black{in Eq. (\ref{s_phi_approx})} is a tractable yet remarkably tight bound, and hence, it introduces a reasonable approximation for the rate coverage probability and offloading gain.
\end{proposition}
\begin{IEEEproof}
The proof of the proposition relies on calculating the concentration bounds for $S_{\Phi_{cpm}^{!}}$. In other words,  we will show that the \ac{RV} $S_{\Phi_{cpm}^{!}}$ concentrates around its mean.  
%
For that purpose, we use Chebyshev's inequality that can be formulated as
 \begin{equation}
\Pb\left(\Big|S_{\Phi_{cpm}^{!}}-\Eb\left[S_{\Phi_{cpm}^{!}}\right]\Big|>a\right) \leq \frac{{\rm Var}\left[S_{\Phi_{cpm}^{!}}\right]}{a^2},	
 \end{equation}
for $a>0$, where ${\rm Var}\left[S_{\Phi_{cpm}^{!}}\right]$ is the variance of $S_{\Phi_{cpm}^{!}}$. We start by calculating the conditional variance ${\rm Var}\left[S_{\Phi_{cpm}^{!}}|h_1\right]$ and mean $\Eb\left[S_{\Phi_{cpm}^{!}}|h_1\right]$ in the next two Lemmas.

\begin{lemma}
\label{ch4:variance-pwr}				
The variance of the signal power received from all active providers except for the nearest device, subject to path loss only, and conditioned on the distance to the nearest active provider $H_1=h_1$, is expressed as
\begin{align}
\label{eq:XXXXX}
{\rm Var}\left[S_{\Phi_{cpm}^{!}}| H_1=h_1\right]  \overset{}{=} c_mp\bar{n} \Gamma\left(-2\alpha+1,\frac{h^2_1}{4\sigma^2}\right),
\end{align}
\end{lemma}
\begin{IEEEproof}
The proof can be found in \App{proof-concentration}.
\end{IEEEproof}

\begin{lemma}
\label{ch4:average-pwr}				
The average over the signal power received from all active content providers except for the nearest one, subject to path loss only, and conditioned on the distance to the nearest active provider $H_1=h_1$, is expressed as
\begin{align}
\label{eq:XXX}
\Eb\left[ S_{\Phi_{cpm}^{!}}\Big| H_1=h_1 \right] = \frac{c_mp\bar{n}}{2\sigma^2}\left[\frac{\exp\left(-\frac{h_1^2}{4 \sigma ^2}\right)}{2 h_1^2}-\frac{\Gamma (0,\frac{h_1^2}{4 \sigma ^2})}{8 \sigma ^2}\right].
\end{align}
\end{lemma}
\begin{IEEEproof}
We can write the conditional mean as
 \begin{align}
 \Eb\left[S_{\Phi_{cpm}^{!}}| H_1=h_1\right] &=  \Eb \Big[ \sum_{\boldsymbol{y}_{0i}\in \Phi_{cpm}^{!}} \lVert \boldsymbol{x}_0 + \boldsymbol{y}_{0i}\rVert^{-\alpha} \Big]
 \nonumber \\
  &\overset{(a)}{=} c_mp\bar{n} \int_{\mathbb{R}^2}^\infty \frac{1}{\lVert \boldsymbol{x}_0 + \boldsymbol{y}_{0i}\rVert^{\alpha}} f_{\boldsymbol{Y}_{0i}}(\boldsymbol{y}_{0i})\dd{\boldsymbol{y}_{0i}}. 
 \end{align}
where (a) follows from the mean and variance for \acp{PPP} \cite[Corollary 4.8] {haenggi2012stochastic}, along with the Gaussian \ac{PPP} assumption $\Phi_{cpm}^{!}$.  
Following the same methodology as in Appendix \ref{app:proof-concentration}, the conditional mean can be directly obtained. Hence, Lemma \ref{ch4:average-pwr} is proven.
\end{IEEEproof}

As an illustrative example, it is reasonable to assume a limited cache-size per device, which triggers $c_m<1$, mean number of devices per cluster $\bar{n}$ from 5 to 10 devices, and small displacement  standard deviation $\sigma$ from $\SI{1}{ m}$ to $\SI{10}{ m}$.
In such setup, we can observe the tightness of our approximation in Fig.~\ref{chebyshev}. Particularly, we plot the term ${\rm Var}\left[S_{\Phi_{cpm}^{!}}\right]/a^2$, measuring how much $S_{\Phi_{cpm}^{!}}$ deviates from its mean, along with the \ac{CDF} of nearest serving distance $F_{H_1}(h_1)$ versus the nearest distance $h_1$.
From the figure, we first note that ${\rm Var}\left[S_{\Phi_{cpm}^{!}}\right]/a^2$ is almost zero when the nearest active provider is farther than $\sigma=\SI{1}{m}$, which happens with high probability (from the \ac{CDF} $F_{H_1}(h_1))$. Moreover, ${\rm Var}\left[S_{\Phi_{cpm}^{!}}\right]/a^2$ is larger than zero when the distance to the nearest active provider is shorter than $\sigma$, which happens with small probability (from the \ac{CDF} $F_{H_1}(h_1))$. This shows that the "mean plus nearest" approximation can yield a remarkably tight bound on $S_{\Phi_{cpm}}$, and correspondingly on $\mathbb{P}_{o}(\boldsymbol{c})$ for scenarios of practical interest. Hence, the proof of Proposition \ref{ch4:concentration} is completed.
\end{IEEEproof}

\begin{figure}[!tbp]	
\vspace{-0.3 cm}
\centering
\includegraphics[width=0.42\textwidth]{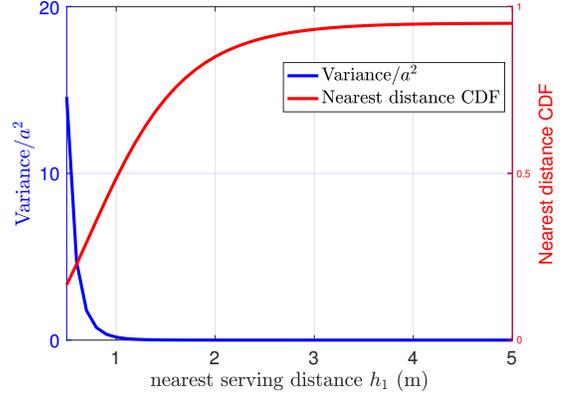} 
\caption {Nearest serving distance CDF $F_{H_1}(h_1)$ (right side y-axis) and ${\rm Var}\left[S_{\Phi_{cpm}^{!}}\right]/a^2$ (left side y-axis) are plotted versus the nearest serving distance $h_1$ ($\sigma=\SI{1}{m}$, $a=1$, $c_m=0.6$, $p=0.5$, $\bar{n}=10)$.}		
\label{chebyshev}
\end{figure}

Piecing everything together, we get a tight approximation on $\mathbb{P}_{o}(\boldsymbol{c})$ as follows. We start with (\ref{eq:rp-exact}) with the substitution $s_{\Phi_{cpm}} \approx h_1^{-\alpha} + \Eb\left[ S_{\Phi_{cpm}^{!}} | H_1=h_1 \right]$, where $\Eb\left[ S_{\Phi_{cpm}^{!}} | H_1=h_1 \right]$ is derived \black{in Eq. (\ref{eq:XXX})}. Then, we proceed by calculating Laplace transform of inter-cluster interference before averaging over the nearest distance $h_1$ using the nearest distance \ac{PDF} $f_{H_1}(h_1)$ \black{in Eq. (\ref{Leibniz})}. The approximated offloading gain is formally characterized in the next corollary. \black{Recall that this approximation is based on replacing the intended signal power subject to path loss only  by two components, namely, the received signal power
from the nearest active provider and the statistical mean over the received power from other content providers.}
\begin{corollary}  
\label{ch4:coro-comp-interference-approx}
A tight approximation of the offloading gain can be calculated from
\begin{align}
\mathbb{P}_{o}^{\approx}(\boldsymbol{c}) &=  \sum_{m=1}^{N_f} q_m\Big(c_m + (1-c_m) \times 
\nonumber \\
\label{ch4:offloading-gain-approx2}
&\int_{h_1=0}^{\infty}   e^{-2\pi \lambda_p\int_{v=0}^{\infty}\big(1 - e^{-p\bar{n}\zeta(v,t)}\big)v\dd{v}}
f_{H_1}(h_1) {\dd h_1}\Big) 
\\
\label{ch4:offloading-gain-approx4}
  &\overset{(a)}{\approx}  \sum_{m=1}^{N_f} q_m\Big(c_m + (1-c_m) p\bar{n}c_m  \int_{h_1=0}^{\infty} \frac{h_1}{2 \sigma^2} 
  \times 
  \nonumber \\
  &
     e^{-2\pi \lambda_p\int_{v=0}^{\infty}\big(1 - e^{-p\bar{n}\zeta(v,t)}\big)v\dd{v}} 
     \times
     \nonumber \\
  &
  \quad\quad \quad\quad \quad\quad e^{-c_mp\bar{n}\big(1-e^{\frac{-h_1^2}{4 \sigma^2}}\big) -\frac{h_1^2}{4 \sigma^2}} {\dd h_1} \Big),
\end{align}
where
\begin{align}
t=\frac{\vartheta/\gamma_d}{h_1^{-\alpha} + \frac{c_mp\bar{n}}{2\sigma^2}\Big[\frac{e^{-\frac{h_1^2}{4 \sigma ^2}}}{2 h_1^2}-\frac{\Gamma (0,\frac{h_1^2}{4 \sigma ^2})}{8 \sigma ^2}\Big]}. 
\end{align}
\end{corollary}
\begin{proof}
The above result follows from the nearest plus mean approximation \black{in Eq. (\ref{s_phi_approx})}, along with the conditional mean expression obtained \black{in Eq. (\ref{eq:XXX})}; (a) follows from the approximated nearest serving distance \ac{PDF} obtained \black{in Eq. (\ref{nearest-1})}.
\end{proof}

Note that the exponential term inside the integral of (\ref{ch4:offloading-gain-approx2}) is a function of $h_1$ since $t=\frac{\vartheta}{\gamma_d s_{\Phi_{cpm}}}$. It is worth mentioning that with such an approximation, replacing $s_{\Phi_{cpm}}$ with its nearest plus conditional mean approximation converts the multi-fold integral over $\boldsymbol{h}_k$ \black{in Eq. (\ref{ch4:offloading-gain-eqn})} to a single integral over $h_1$. Furthermore, the effect of having a random number of active providers is implicitly involved in the nearest distance \ac{PDF} as well as in the conditional mean term. This explains why the condition of having $k$ active content providers in the representative cluster no longer exists in the approximated rate coverage probability expression.

 \begin{figure}[!tb]	
 \vspace{-0.3 cm}
\centering
\includegraphics[width=0.42\textwidth]{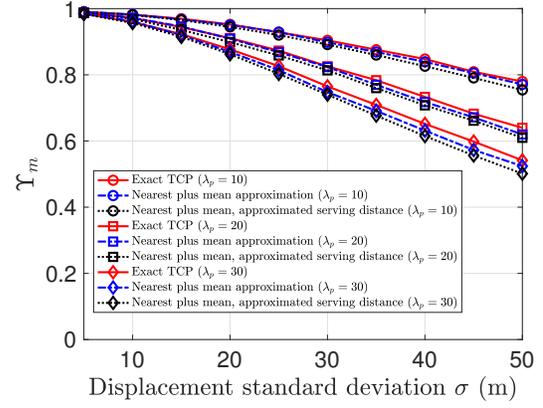}	   
\caption {The derived approximations of $\Upsilon_m$ \black{in Eq. (\ref{ch4:offloading-gain-approx2})} and (\ref{ch4:offloading-gain-approx4}) are plotted versus the displacement standard deviation $\sigma$ for various parent \ac{PPP} density $\lambda_p$ ($\bar{n}=20$, $p=0.5$, $c_m=0.5$). "Nearest plus mean approximation" in the legend refers to the performance based on the exact nearest serving distance \ac{PDF} \black{in Eq. (\ref{Leibniz})}.}
\label{fig:approximation1-2}
\vspace{-0.6 cm}
\end{figure}
In \Fig{approximation1-2}, we plot the obtained formulas for the rate coverage probability $\Upsilon_m$ \black{in Eq. (\ref{ch4:offloading-gain-approx2}) and Eq. (\ref{ch4:offloading-gain-approx4})} versus the displacement variance $\sigma$ for various density of clusters $\lambda_p$. The derived bound is remarkably tight for the practical values of $\lambda_p$ and $\sigma$. It starts to slightly diverge only for a system with considerably large $\lambda_p$ and $\sigma$. It is also clear that the approximated expression of the nearest distance \ac{PDF} \black{in Eq. (\ref{nearest-1})} bounds that \black{in Eq. (\ref{Leibniz})} very tightly. 

It is worth mentioning that the importance of the obtained approximation $\mathbb{P}_{o}^{\approx}(\boldsymbol{c})$ \black{in Eq. (\ref{ch4:offloading-gain-approx2})} is two-fold. Firstly, it provides an easy-to-compute approximation for the rate coverage probability, as shown in \Fig{approximation1-2}, and correspondingly, the offloading gain. Secondly, it allows us to solve for an optimized caching probability, by maximizing the approximated offloading gain, as will be clear shortly. In this regard, it is quite important to examine the achievable performance when a requested content is downloaded only from the nearest active provider. We refer to this as \emph{\ac{NCP} transmission scheme}, where its corresponding  rate coverage probability is characterized in the next corollary. 		
\begin{corollary}
\label{corollary-near}
The rate coverage probability for the \ac{NCP} transmission scheme is expressed as
\begin{align}
&\Upsilon_m=   
 \int_{0}^{\infty}   e^{-2\pi \lambda_p\int_{v=0}^{\infty}\big(1 - e^{-\bar{n}\zeta(v,t)}\big)v\dd{v}} f_{H_1}(h_1) {\dd h_1}     
\nonumber \\
  & \approx   \frac{c_mp\bar{n}}{2 \sigma^2}
\int_{0}^{\infty} h_1   e^{-2\pi \lambda_p\int_{v=0}^{\infty}\big(1 - e^{-\bar{n}\zeta(v,t)}\big)v\dd{v}} 
\times
\nonumber \\
  & \quad\quad\quad\quad\quad\quad\quad\quad\quad\quad\quad\quad\quad\quad\quad
e^{-c_mp\bar{n}\big(1-e^{(\frac{-h_1^2}{4 \sigma^2})}\big) -\frac{h_1^2}{4 \sigma^2}} {\dd h_1}, 
\nonumber 
\end{align}
where $t=\frac{\vartheta}{\gamma_d h_1^{-\alpha}}$, and $\zeta(v,t)$ is defined in Lemma \ref{ch4:comp-interference}.
\end{corollary}

Having obtained a lower bound and approximation of the rate coverage probability, next, we employ the obtained results to compute optimized caching probabilities $\boldsymbol{c}$ in order to maximize the offloading gain. 

\vspace{-0.3 cm}
\section{Optimized Caching Probabilities}
We first formulate and solve the offloading gain maximization problem based on the mean plus nearest approximation \black{in Eq. (\ref{ch4:offloading-gain-approx4})}. Then we pursue another approach, based on the derived lower bound \black{in Eq. (\ref{ch4:lower-bound-offload-gain})}, to obtain a low complexity solution.

\vspace{-0.4 cm}
\subsection{Optimized Caching Based on the approximation}
Here, we aim at maximizing the approximated offloading gain obtained in Corollary \ref{ch4:coro-comp-interference-approx}. We formulate the offloading gain maximization problem as 
\begin{align}
\label{optimize_eqn_p}
\textbf{P1:}		\quad &\underset{\boldsymbol{c}}{\text{max}} \quad \mathbb{P}_{o}^{\approx}(\boldsymbol{c}) 
\\
\label{const1}
&\textrm{s.t.}\quad  \sum_{n=1}^{N_f} c_m = M, \quad   c_m \in [ 0, 1]  		
\end{align}
Since the integral in the approximated offloading gain expression \black{in Eq. (\ref{ch4:offloading-gain-approx2})} depends on the caching probability $c_m$, and $c_m$ exists as a complex exponential term in the nearest serving distance \ac{PDF} $f_{H_1}(h_1)$, it is hard to analytically characterize (e.g., show concavity of) the objective function or find a tractable expression for the caching probability $c_m$. In order to tackle this, similar to \cite{cache_schedule}, we introduce a $\boldsymbol{c}$-independent integral by substituting the caching probability with an arbitrary caching probability $\boldsymbol{c}^0$. Denoting this $\boldsymbol{c}$-independent integral by $I^{\boldsymbol{c}^0}$, it can be easily verified that 
\begin{align}
\mathbb{P}_{o}^{\approx^{\boldsymbol{c}^0}}(\boldsymbol{c})=\sum_{m=1}^{N_f} q_m\Big(c_m + (1-c_m)p\bar{n}c_m I^{\boldsymbol{c}^0}\Big),
\end{align}
 is concave in $\boldsymbol{c}$. From the \ac{KKT} conditions, the optimized caching probability maximizing 
 $\mathbb{P}_{o}^{\approx^{\boldsymbol{c}^0}}(\boldsymbol{c})$ under the constraint (\ref{const1}) is given by 
 \begin{align}
 \label{convex-approx}
c_m^*  =     \Big[0.5  -  \frac{v^{*}-q_m}{2q_mp\bar{n} I^{\boldsymbol{c}^0}} \Big]^+,
\end{align}
where $v^*$ satisfies the maximum cache constraint $\sum_{i=1}^{N_f} \Big[0.5  -  \frac{v^{*}-q_m}{2q_mp\bar{n} I^{\boldsymbol{c}^0}} \Big]^+ = M$, and $[x]^+ = {\rm max}(x, 0)$. For the arbitrary caching probability $\boldsymbol{c}^0$, a locally optimal caching probability can be adopted, which can be computed via, e.g., interior point method \cite{boyd2004convex}. Nonetheless, we can increase the probability of finding the optimal solution of \textbf{P1} by using the interior point method with multiple random initial values and then picking the solution with the  highest offloading gain. 
We refer to this caching probability \black{in Eq. (\ref{convex-approx})} as the solution from convex approximation. However, to obtain a caching policy of lower complexity, we maximize a special case of the  lower bound  $\mathbb{P}_{o}^{\sim}(\boldsymbol{c})$ in the sequel.

\vspace{-0.0 cm}
\subsection{Optimized Caching Based on the Lower Bound}
Although $\mathbb{P}_{o}^{\sim}(\boldsymbol{c})$ characterized in Theorem \ref{ch4:comp-interference-approx} is simpler to compute compared to $\mathbb{P}_{o}(\boldsymbol{c})$, it is still challenging to obtain the optimal caching probability due to the summation and multi-fold integration \black{in Eq. (\ref{ch4:lower-bound-offload-gain})}. Similar to \cite{chae2017content}, we consider a special case when downloading content from one active content provider, for which the offloading gain maximization problem turns out to be convex. 

\black{\textbf{\emph{One Content Provider:}}}		
Next, we solve for the optimized caching probability when considering one serving content provider (instead of $k$ \black{in Eq. (\ref{ch4:lower-bound-offload-gain})}). Starting from (\ref{ch4:lower-bound-offload-gain}) with $t = \frac{\vartheta h^{\alpha}}{\gamma_d}$ for one serving provider, and the void probability of a Poisson \ac{RV}, we get
\begin{align}
\label{rate-cov-one}
&\Upsilon_m= \big(1 - \frac{(p\bar{n}c_m)^0}{0!}e^{-c_mp\bar{n}}\big) \times
\nonumber \\
&
\int_{h=0}^{\infty} e^{-\pi p\bar{n}\lambda_p (\vartheta h^{\alpha})^{2/\alpha} \Gamma(1 + 2/\alpha)\Gamma(1 - 2/\alpha)}
 \frac{h}{2\sigma^2}e^{-\frac{h^2}{4\sigma^2}} \dd{h}.
\end{align} 
Solving the integral \black{in Eq. (\ref{rate-cov-one})}, and substituting \black{in Eq. (\ref{ch4:lower-bound-offload-gain})}, we get
$\mathbb{P}_{o}^{\sim1}(\boldsymbol{c})$ written as 
\begin{align}
\label{ch4:lower-bound-closed-form}
\mathbb{P}_{o}^{\sim1}(\boldsymbol{c}) & \overset{}{=}\sum_{m=1}^{N_f} q_m\Big(c_m + \big(1-c_m\big) 
\big(1 - e^{-c_mp\bar{n}}\big)
   \frac{1}{\mathcal{Z}(\vartheta,\alpha,\sigma)}\Big),
\end{align}   
where $\mathcal{Z}(\vartheta,\alpha,\sigma) = 4\sigma^2\pi p\bar{n}\lambda_p \vartheta^{2/\alpha}\Gamma(1 + 2/\alpha)\Gamma(1 - 2/\alpha)+ 1$. 
Hence, the optimized caching probability can be computed by solving the following problem
\begin{align}
\label{optimize_eqn_p}
\textbf{P2:}		\quad &\underset{\boldsymbol{c}}{\text{max}} \quad \mathbb{P}_{o}^{\sim1}(\boldsymbol{c}) \\
\label{const110}
&\textrm{s.t.}\quad  \sum_{n=1}^{N_f} c_m = M, \quad   c_m \in [ 0, 1] 
\end{align}

The optimal solution for \textbf{P2} is formulated in the following Lemma.	
\begin{lemma}
\label{ch4:offload-concave}
The lower bound on the offloading gain $\mathbb{P}_{o}^{\sim1}(\boldsymbol{c})$ \black{in Eq. (\ref{ch4:lower-bound-closed-form})} is concave w.r.t.  the caching probability, and the optimal probabilistic caching $\underline{\boldsymbol{c}}^*$ for \textbf{P2} is given by
     \[
    \underline{c_{m}}^{*}=\left\{
                \begin{array}{ll}
                  1 \quad\quad\quad , v^{*} <   q_m -\frac{q_m(1-e^{-p\bar{n}})}{\mathcal{Z}}\\ 
                  0   \quad\quad\quad, v^{*} >   q_m + \frac{p\bar{n}q_m}{\mathcal{Z}}\\
                 \psi(v^{*}) \quad, {\rm otherwise},
                \end{array}   
              \right.
  \]
where $\psi(v^{*})$ is the solution of
\begin{align}
v^{*} =   q_m  +  \frac{q_m}{\mathcal{Z}}\big(p\bar{n}(1-\underline{c_{m}}^{*})e^{-p\bar{n}\underline{c_{m}}^{*}} - (1 - e^{-p\bar{n}\underline{c_{m}}^{*}})\big),
\nonumber 
\end{align}
 that satisfies $\sum_{m=1}^{N_f} \underline{c_{m}}^{*}=M$, and $\mathcal{Z} = \mathcal{Z}(\vartheta,\alpha,\sigma)$ for ease of presentation.

\end{lemma}

\begin{proof}
It is easy to show the concavity of the objective function $\mathbb{P}_{o}^{\sim1}(\boldsymbol{c})$ by confirming that Hessian matrix w.r.t. the caching variables is negative semi-definite. Also, the constraints are linear, which imply that the necessity and sufficiency conditions for optimality exist. The detailed proof of finding $\underline{\boldsymbol{c}}^*$ is omitted for brevity.
\end{proof}

\black{Noticeably, the offloading gain for this special case, specifically when considering one content provider, resembles a strict lower bound on the exact offloading gain \cite{chae2017content}. However, given the complexity of the original optimization problem,  we sought a suboptimal caching solution by optimizing a convex lower bound. Similar approaches are used in the literature, especially when dealing with complex optimization problems like ours, see, e.g., \cite{chae2017content,cache_schedule}, and \cite{7417164}. When substituting this suboptimal solution in  \black{in Eq. (\ref{ch4:offloading-gain-eqn})}, it provides useful insights into the system design and also attains considerable performance improvements over traditional caching schemes, as quantified in Section \ref{Numerical-Results}.} 


%
\vspace{-0.3 cm}
\section{Numerical Results}
\label{Numerical-Results}

\begin{table}[!t]		
\vspace{-0.0 cm}
\caption{Simulation Parameters} 
\centering 
\scalebox{0.9}{
\begin{tabular}{|c| c| c| c| c | c|} 
\hline 
Description & Parameter & Value\\ [0.5ex] 
\hline 
Path loss exponent& $\alpha$& 4\\ \hline 
$\sir$ threshold&$\vartheta$&\SI{0}{\deci\bel}\\ \hline 
Density of clusters&$\lambda_p$&$\SI{30}{km^{-2}}$\\ \hline
Displacement standard deviation&$\sigma$&$\SI{30}{m}$\\ \hline
Average number of devices per cluster&$\bar{n}$&$6$\\ \hline
Library size& $N_f$& 12\\  \hline 
Device cache size& $M$& 2\\  \hline 
Access probability& $p$& 0.5\\  
\hline 
\end{tabular}}
\label{ch4:table:sim-parameter} 
\vspace{-0.4cm}
\end{table}

In this section, we evaluate the performance of our proposed cache-assisted \ac{CoMP} transmissions for clustered D2D networks.  Unless otherwise specified, results are obtained for the parameters shown in Table 1, which represent typical values used in many previous works.\footnote{\black{In Table 1, we assume relatively small device cache and content library sizes. These values, which are very close to those used in  \cite{chae2017content,8374852,cache_schedule,amer2019optimizing}, and \cite{7417164}, are reasonable in the study of communication and caching aspects of D2D content delivery networks. Other works in the literature, e.g., \cite{golrezaei2014base}, considered a much larger size of file library, however, their objective was to conduct the scaling analysis of caching networks.}} We refer to both solutions based on convex approximation \black{in Eq. (\ref{convex-approx})} and the suboptimal caching of Lemma 5 as optimized \ac{PC}. We first verify the accuracy of the derived bound and approximation of the offloading gain. Then, we compare the achievable performance of our proposed \ac{PC} and CoMP transmission, with conventional caching and transmission schemes.

\subsection{Exact Offloading Gain Versus Approximation and Lower Bound}

\begin{figure}[tb]	
\centering
\includegraphics[width=0.45\textwidth]{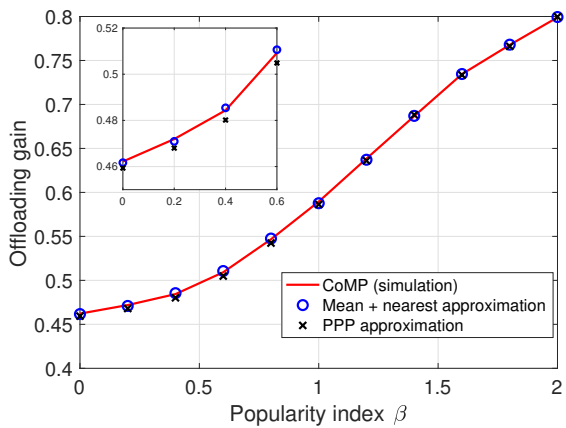}		
\caption {The exact offloading gain (simulation) based on \ac{CoMP} transmissions is compared to PPP-based lower bound ($\mathbb{P}_{o}^{\sim}(\boldsymbol{c}))$, and mean plus nearest-based approximation ($\mathbb{P}_{o}^{\approx}(\boldsymbol{c}))$, versus the popularity of files $\beta$ under the Zipf caching scheme.}
\label{offloading_gain_vs_beta_CoMP}				
\vspace{-0.0 cm}		
\end{figure}

Fig.~\ref{offloading_gain_vs_beta_CoMP} verifies the accuracy of the obtained lower bound and approximation of the offloading gain \black{in Eq. (\ref{ch4:lower-bound-offload-gain}) and Eq. (\ref{ch4:offloading-gain-approx2})}, respectively. It is clear that both the derived lower bound and approximation are tight to the exact offloading gain. Moreover, we observe that averaging over the request and caching probabilities, i.e., $q_m$ and $c_m$, makes the adopted lower bound and approximation of the offloading gain tighter than those for the rate coverage probability $\Upsilon_m$ shown earlier in Fig.~\ref{fig:lower-bound} and Fig.~\ref{fig:approximation1-2}, respectively. As shown in  Fig.~\ref{offloading_gain_vs_beta_CoMP}, the offloading gain monotonically increases with the popularity of files $\beta$. This is because when $\beta$ is large, only a small portion of content undergoes most of the demand, which can be cached among the cluster devices.

 \vspace{-0.3 cm}
\subsection{Comparison with Other Caching Schemes}

\begin{figure}[!t]	
\vspace{-0.3 cm}
\centering
\includegraphics[width=0.40\textwidth]{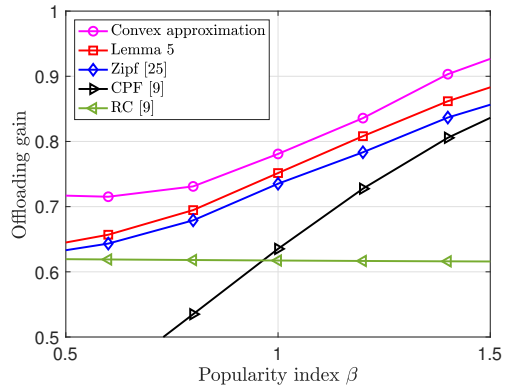}		
\caption {The offloading gain versus the popularity of files $\beta$ under different caching schemes 
(\black{$N_f=40$, $M=8$}).}
\label{offloading_gain_vs_beta}
\vspace{-0.0 cm}	
\end{figure}						

Fig.~\ref{offloading_gain_vs_beta} compares the offloading gain of the proposed \ac{PC} with other benchmark schemes, namely, Zipf caching (Zipf), \ac{CPF}, and \ac{RC} against the popularity of files $\beta$. \black{For \ac{CPF},  the $M$-most popular files are cached among each cluster device, see, e.g. \cite{8412262}. Similarly, for \ac{RC}, contents to be cached are randomly chosen according to a uniform distribution irrespective of the popularity as in \cite{8412262}, while for Zipf caching, contents are probabilistically cached according to their popularity as given in (4), see, e.g., \cite{breslau1999web}.} 
%
%
Moreover, in Fig.~\ref{offloading_gain_vs_beta}, "convex approximation" refers to the caching solution characterized \black{in Eq. (\ref{convex-approx})}, whereas "Lemma 5" refers to the caching solution characterized in Lemma 5. All the caching schemes are evaluated under \ac{CoMP} transmissions. We first see that the offloading gains under the optimized \ac{PC} schemes outperform conventional caching schemes.  Moreover, the \ac{PC} based on convex approximation \black{in Eq. (\ref{convex-approx})} is superior to the suboptimal solution of Lemma \ref{ch4:offload-concave}.   
%
%
As $\beta$ increases, except for \ac{RC}, the offloading gain increases and gradually, the optimized \ac{PC}, Zipf, and \ac{CPF} schemes tend to achieve the same performance. This shows that when a small portion of files becomes highly demanded, i.e., for higher $\beta$, the optimal caching probability is attained via caching popular files among all cluster devices.


\begin{figure} [!t]
    \centering
    {\hspace{-0.4 cm}
    \subfigure[Case one ($\sigma=\SI{10}{m}$, $\lambda_p=\SI{10}{km^{-2}}$)]
    {
        \includegraphics[width=2.1in]{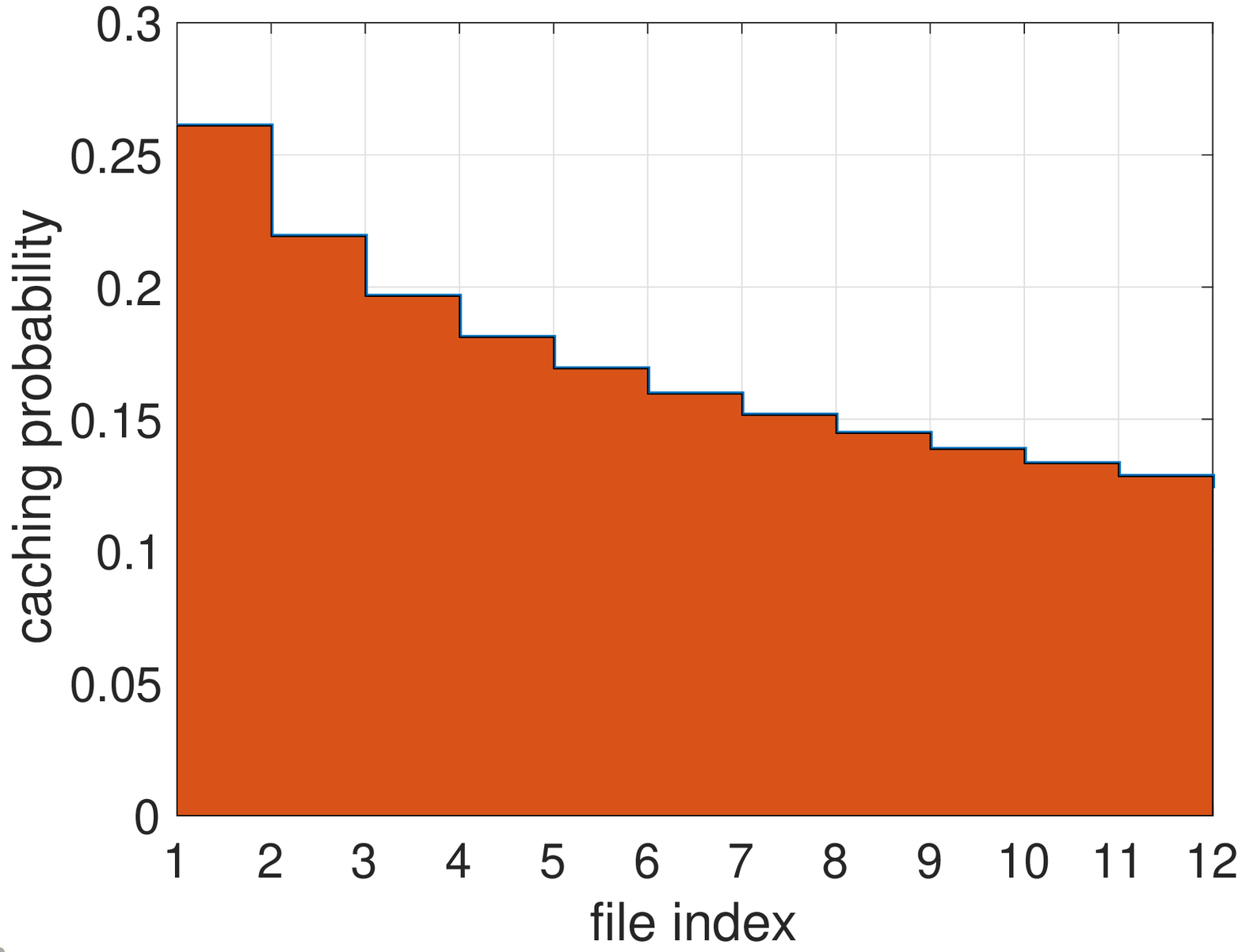}		
        \label{case_one}
    }}
    \subfigure[Case two ($\sigma=\SI{30}{m}$, $\lambda_p=\SI{30}{km^{-2}}$)]
    {
        \includegraphics[width=2.1in]{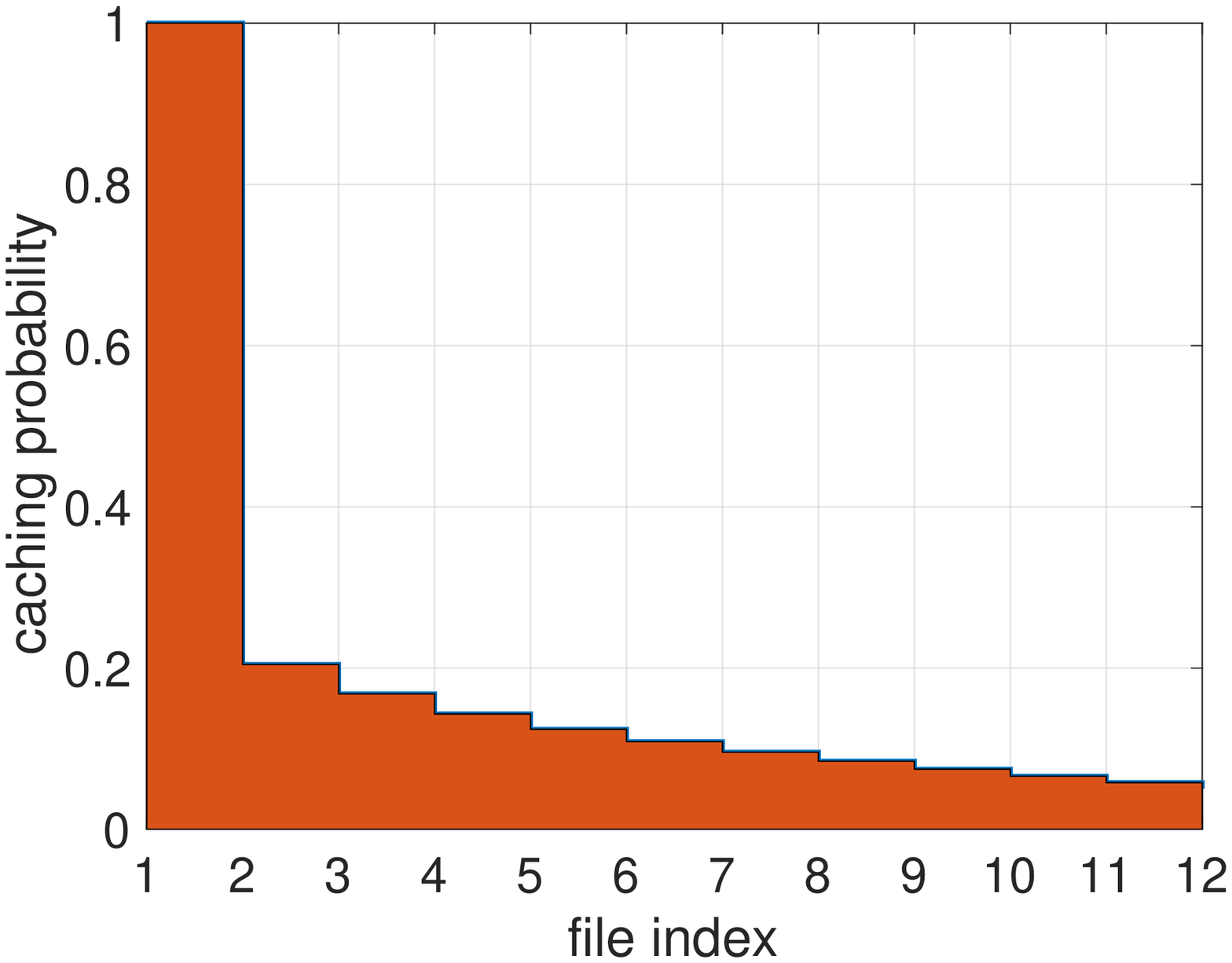}		
        \label{case_two}		
    }
    \subfigure[Case three ($\sigma=\SI{50}{m}$, $\lambda_p=\SI{40}{km^{-2}}$)]
    {
        \includegraphics[width=2.1in]{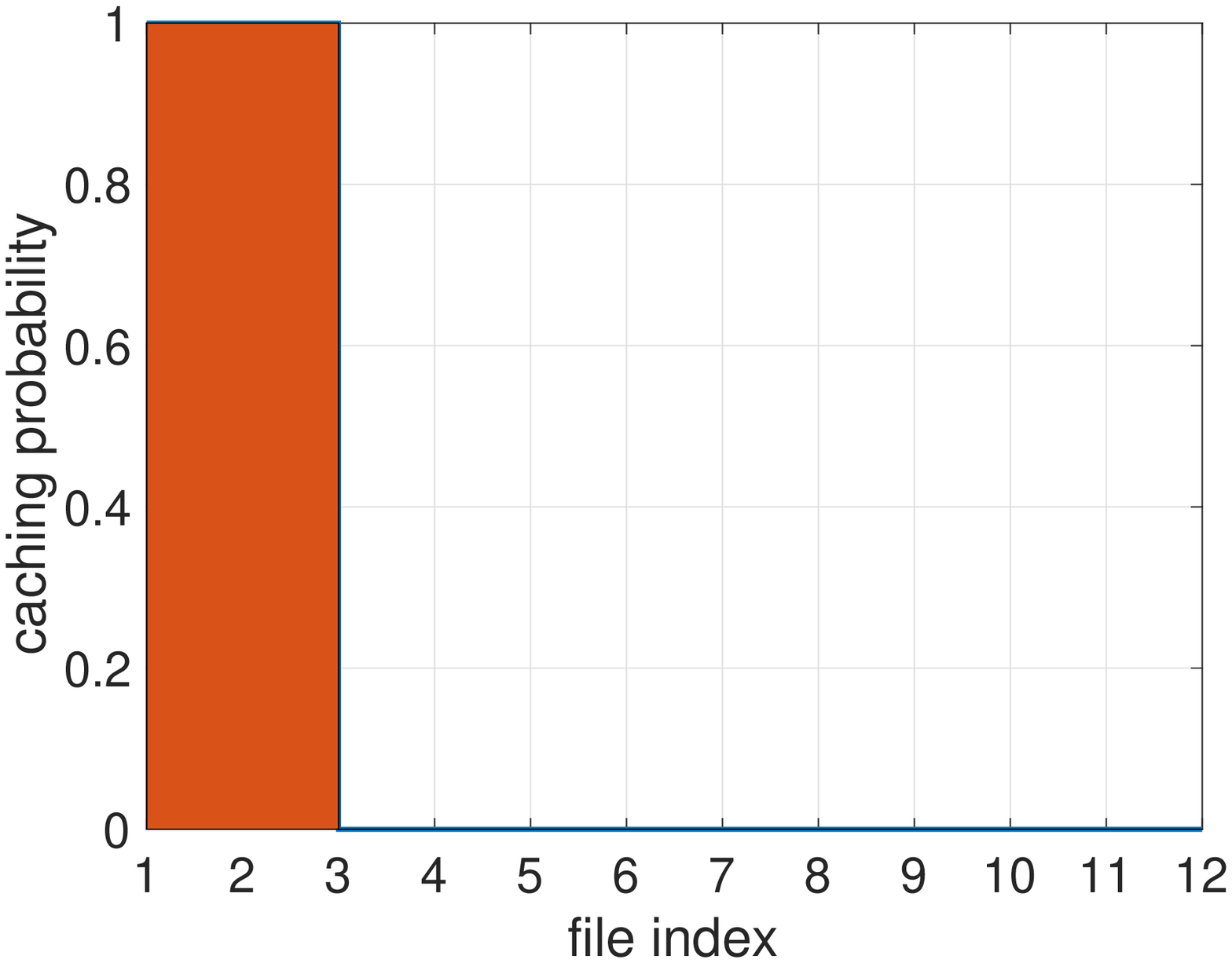}
        \label{case_three}
    }
    \caption{Histogram of the caching solution $\underline{\boldsymbol{c}}^*$ is plotted for different network geometries ($\beta=0.4$).}
    \label{histogram_b_i}
    \vspace{-0.1cm}
\end{figure} 		

To show the prominent effect of the network geometry on the optimized caching probability, we plot the histograms of the solution of Lemma \ref{ch4:offload-concave} for three different cases in Fig.~\ref{histogram_b_i}. These three cases are ranging from a sparse network (small $\sigma$ and $\lambda_p$), a relatively dense network (medium $\sigma$ and $\lambda_p$), and a highly dense network (large $\sigma$ and $\lambda_p$). Note that smaller $\sigma$ results in higher desired signal power, while smaller $\lambda_p$ yields lower inter-cluster interference power as clusters become sparser. 
The first case in Fig.~\ref{case_one} represents a sparse system with small values of $\lambda_p$  and $\sigma$, i.e., sufficiently good transmission conditions. 
In this case, we see that the optimized caching probability tends to be more uniform taking advantage of hitting a large number of files while being served in favorable transmission conditions. 
The second case in Fig.~\ref{case_two} represents a system with relatively good transmission conditions, i.e., medium values of $\sigma$ and $\lambda_p$. It is clear from the histogram that the optimized caching solution tends to be more skewed than in the first case. 
The third case in Fig.~\ref{case_three} is then for a highly dense network with large values of both $\sigma$ and $\lambda_p$. Clearly, the caching probability tends to be very skewed, which implies that caching popular files is an appropriate choice for such a highly dense network, i.e., a network with poor transmission conditions.  Summing up, the results in Fig.~\ref{histogram_b_i} reveal interesting dependence of the optimized caching probability in Lemma \ref{ch4:offload-concave} on the network geometry.

\vspace{-0.3 cm}
\subsection{Comparison with Other Transmission Schemes}

 \begin{figure}[!tb]
 \vspace{-0.3 cm}
	\begin{center}
		\includegraphics[width=0.45\textwidth]{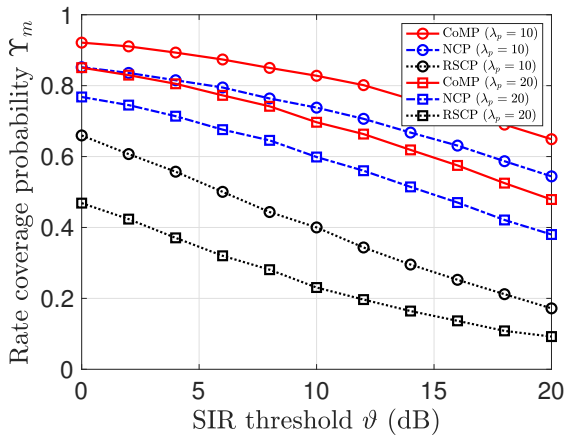}			
		\caption {The rate coverage probability versus the \ac{SIR} threshold $\vartheta$ 
		for different transmission schemes \black{($\bar{n}=20$, $c_m=0.5$)}.}			
		\label{cov_prob_compare_vs_theta}
	\end{center}
\vspace{-0.5cm}			
\end{figure}
To quantify how much CoMP transmission can improve the achievable performance, we here compare the rate coverage probability  $\Upsilon_m$ for three transmission schemes, namely, \ac{CoMP}, \ac{NCP}, and \ac{RSCP}. Recall that for the \ac{NCP} scheme, the requested content is served from the the nearest active provider to the content client within the same cluster while for the \ac{RSCP} scheme, an active provider is chosen at random to serve the desired content.  Firstly, Fig.~\ref{cov_prob_compare_vs_theta} plots the exact rate coverage probability versus  \ac{SIR} threshold $\vartheta$ for different density of clusters $\lambda_p$. Intuitively, \ac{CoMP} transmissions achieves higher rate coverage probability than those of the other schemes. In particular, at high \ac{SIR} threshold, allowing \ac{CoMP} transmissions can provide up to $300\%$ improvement 
in the rate coverage probability compared to the \ac{RSCP} scheme. Moreover, for all schemes, the rate coverage probability is seen to decrease as $\lambda_p$ increases since higher interference power is encountered at the typical client when D2D clusters are denser. 

 \begin{figure}[!t]
 \vspace{-0.3 cm}
	\begin{center}
		\includegraphics[width=0.45\textwidth]{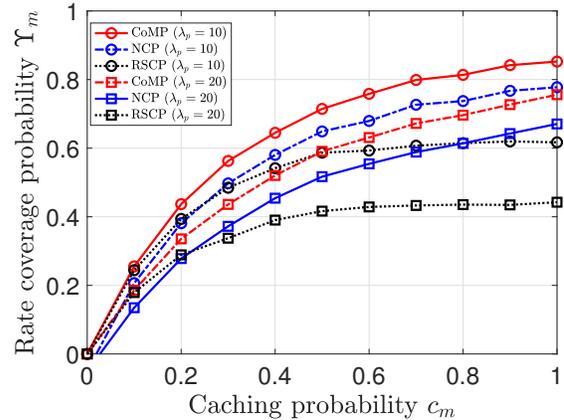}			
		\caption {The rate coverage probability versus the caching probability $c_m$ for different transmission schemes \black{($\bar{n}=10$, $\vartheta=\SI{5}{dB}$)}.}			
		\label{cm_threshold}
	\end{center}
\vspace{-0.3cm}		
\end{figure}	
%
In light of this comparison, Fig.~\ref{cm_threshold} plots  the rate coverage probability  against the caching probability $c_m$ for the three transmission schemes. As shown in Fig.~\ref{cm_threshold}, as the content availability increases, i.e., higher $c_m$, the rate coverage probability improves. Besides, Fig.~\ref{cm_threshold} illustrates that while the rate coverage probability for the \ac{RSCP} transmission scheme tends to flatten when $c_m$ further increases, it continues increasing for  \ac{CoMP} and \ac{NCP} transmission schemes. This is attributed to the fact that, for the \ac{NCP} scheme, the serving distance is more likely to decrease with the increase of $c_m$, and hence the corresponding performance improves. Similarly, for the \ac{CoMP} scheme, the transmission diversity improves with $c_m$ since the average number of active and caching providers increases.

 \begin{figure}[!t]
 \vspace{-0.3 cm}
	\begin{center}
		\includegraphics[width=0.45\textwidth]{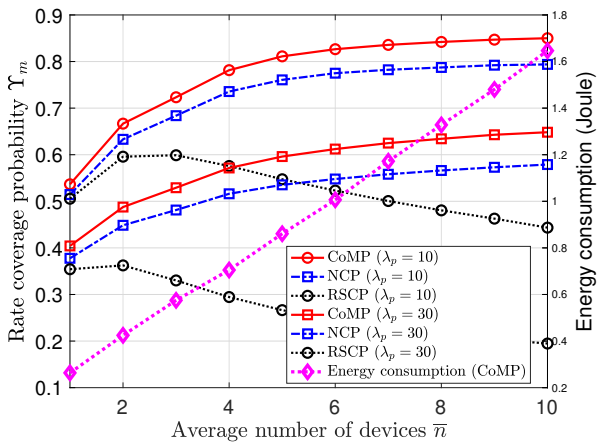}			
		\caption {The rate coverage probability (left hand side y-axis) and the amount of consumed energy per content request in one cluster (right hand side y-axis) versus the average number of devices per cluster $\bar{n}$ ($\vartheta=\SI{8}{dB}$, $c_m=1$, $P_d=\SI{20}{dBm}$, $\bar{S}=\SI{5}{MBytes}$, $W=\SI{5}{MHz}$).}	
		\label{sir_threshold}
	\end{center}
\vspace{-0.5cm}
\end{figure}	
%

Finally, Fig.~\ref{sir_threshold} shows the effect of the average number of devices per cluster $\bar{n}$ on the rate coverage probability $\Upsilon_m$ (left hand side y-axis). \black{Fig.~\ref{sir_threshold} also presents the amount of consumed energy per content request in one cluster (right hand side y-axis)}. Fig.~\ref{sir_threshold} first shows that as $\bar{n}$ increases, the rate coverage probability for CoMP and \ac{NCP} schemes increases. This is due to the fact that, for the \ac{NCP} scheme, the nearest serving distance is more likely to decrease for congested clusters. Moreover, for the CoMP scheme, the transmission diversity also improves with $\bar{n}$. However, for the \ac{RSCP} scheme, the  rate coverage probability first increases driven by the increasing probability of finding a requested content within the local cluster. Then, the rate coverage probability turns to decrease as $\bar{n}$ further increases. \black{This degradation is attributed to the fact that while the requested content can be found within the local cluster with high probability for higher $\bar{n}$, the effect of inter-cluster interference also grows with $\bar{n}$. The latter then becomes the dominant factor that drives the rate coverage probability $\Upsilon_m$ down. Please notice that as $\bar{n}$ increases, the number of collaborative devices in other remote clusters correspondingly increases, which yields higher interference power.} Furthermore, the rate coverage probability is shown to decrease for all schemes when $\lambda_p$ increases driven by the growing effects of interference. Besides, while the performance of \ac{RSCP} is very sensitive to the density of clusters, especially when the average number of devices is high, cooperative transmission attains substantially better performance. We hence conclude that for such highly dense D2D networks and adverse interference conditions, cooperative transmission becomes more appealing. 
\black{From Fig.~\ref{sir_threshold}, we also see that the consumed energy monotonically increases with the average number of devices per cluster $\bar{n}$. This is attributed to the fact as the number of content providers increases with $\bar{n}$, more energy will be consumed to serve a content. The average energy consumption per content request in one cluster is calculated from $E_m = q_m c_m \bar{n} \bar{E}$, with $\bar{E}=\frac{P_d\bar{S}}{R_d}$ being the consumed energy per requested content per content provider; $P_d$ denotes the device transmission power, $\bar{S}$ is the average content size, and $R_d$ is the \ac{D2D} average transmission rate that is calculated numerically from $R_d=W {\log}_2(1+\vartheta)\Upsilon_m$ .}

\vspace{-0.5 cm}
\section{Conclusion}
In this paper, we have conducted performance analysis and content placement optimization for cache-assisted \ac{CoMP} transmissions in clustered \ac{D2D} networks. In particular, we have characterized the rate coverage probability and offloading gain as functions of the network parameters, namely, the density of clusters, average number of devices per cluster, and the content popularity and placement schemes. Then, we have sought simple yet tight lower bound and approximation of the rate coverage probability and offloading gain. Based on the obtained results, \black{we have shown that} the inter-cluster interference of a \ac{TCP} is upper bounded by that of a \ac{PPP} of the same intensity. Moreover, we have formulated the corresponding offloading gain maximization problem and obtained optimized caching probabilities based on the proposed lower bound and approximation. Results showed that allowing \ac{CoMP} transmissions can attain up to $300\%$ improvement in the rate coverage probability compared to the \ac{RSCP} scheme. Finally, we conclude by showing that the proposed optimized \ac{PC} results in a considerable improvement of the offloading gain over conventional caching schemes. 
%

%
%

\ifCLASSOPTIONcaptionsoff
  \newpage
\fi


%

%
%

\vspace{-0.3 cm}
\bibliographystyle{IEEEtran}
\bibliography{public_files/bibliography}

%

\vspace{-0.4 cm}

\begin{IEEEbiography}[{\includegraphics[width=1in,height=1.25in,clip,keepaspectratio]{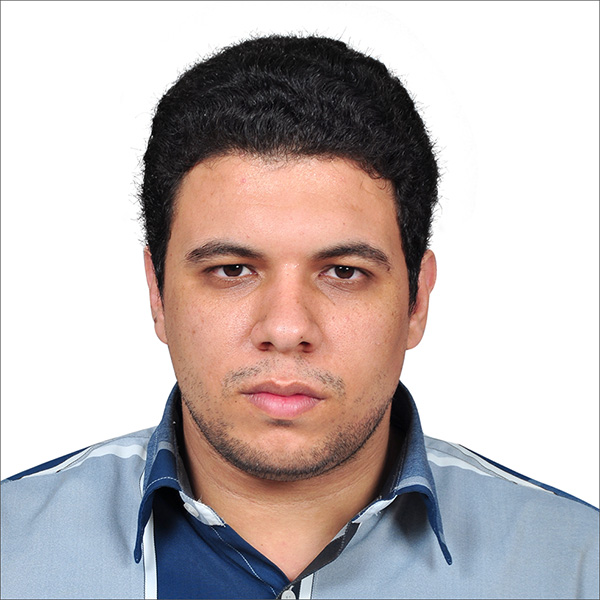}}]{Ramy Amer} received the MSc degree in Electrical Engineering from Alexandria University, Egypt, in 2016. Since 2016, he is pursuing his PhD degree in wireless communication, especially, wireless caching, under the supervision of Dr Nicola Marchetti at CONNECT centre, Trinity College Dublin, Ireland. His research interests include cross-layer design, Cognitive Radio, Energy Harvesting, Stochastic Geometry, and Reinforcement Learning. Prior to CONNECT, he was also an assistant lecturer for the switching department, National Telecommunication Institute of Egypt (NTI), Cairo, Egypt, where he conducted professional training both at the national and international levels. He was a visiting scholar at Wireless at Virginia Tech, with Prof Walid Saad's group from September 2018 to March 2019. He is the recipient of best paper award from IFIP NTMS in 2019 and the IEEE student travel grant from WCNC in 2019. He is recognized as an Exemplary Reviewer by the IEEE TRANSACTIONS ON COMMUNICATIONS 2019. 
He is also certified as a Cisco instructor and has other Cisco data and voice certificates. He worked as a part-time instructor for many national and international training centres, e.g., New-Horizon Egypt and Fast-Lane KSA. 
\end{IEEEbiography}

\vspace{-0.4 cm}
\begin{IEEEbiography}[{\includegraphics[width=1in,height=1.25in,clip,keepaspectratio]{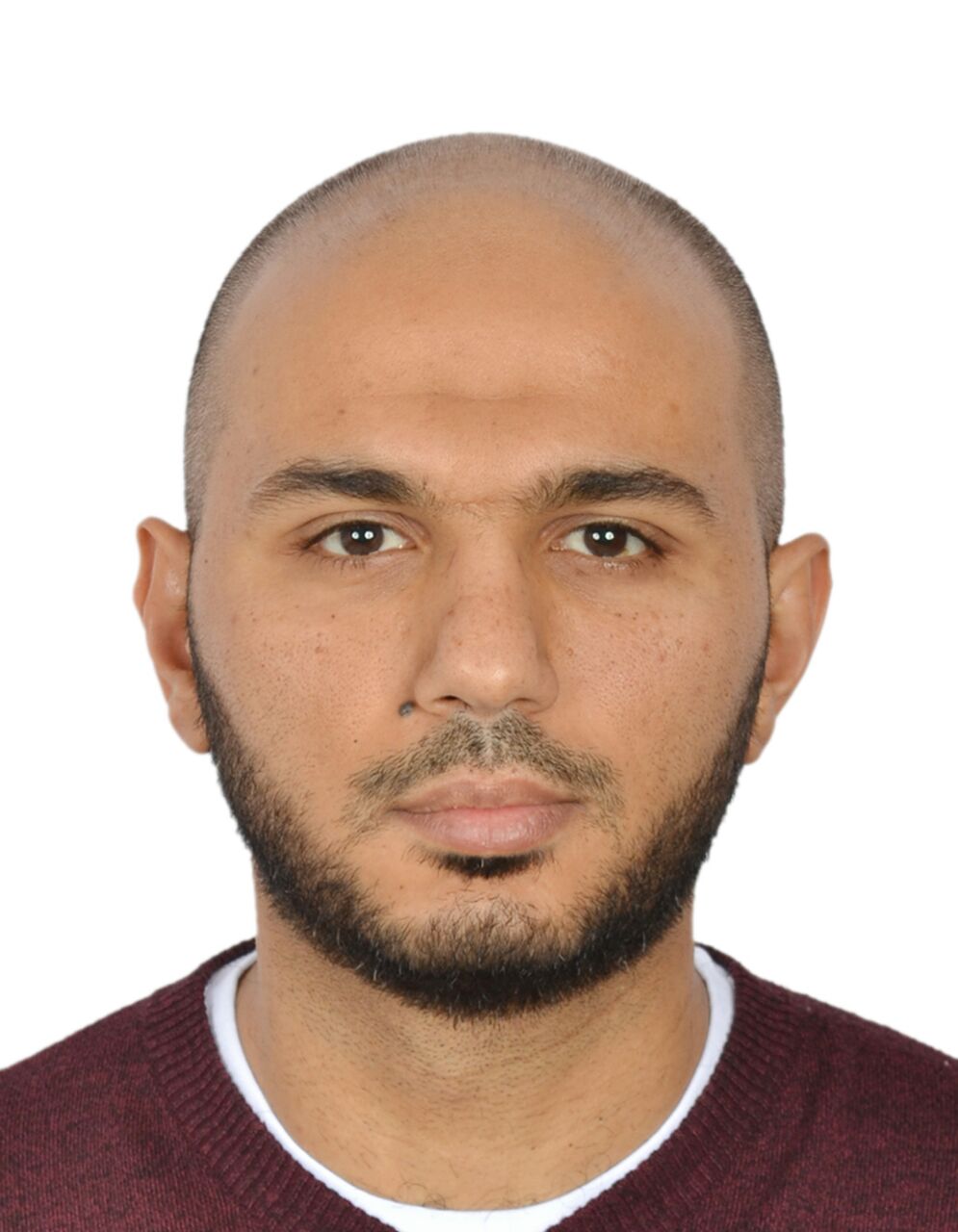}}]{Hesham ElSawy}(S'10 -- M'14 -- SM'17) is an assistant professor at King Fahd University of Petroleum and Minerals (KFUPM), Saudi Arabia. Prior to that, he was a postdoctoral Fellow at King Abdullah University of Science and Technology (KAUST), Saudi Arabia, a research assistant at TRTech, Winnipeg, MB, Canada, and a telecommunication engineer at the National Telecommunication Institute, Egypt. Dr. ElSawy received his Ph.D. degree in Electrical Engineering from the University of Manitoba, Canada, in 2014, where he received several academic awards, including the NSERC Industrial Postgraduate Scholarship during the period of 2010-2013, and the TRTech Graduate Students Fellowship in the period of 2010-2014. He co-authored three award-winning papers that are recognized by the IEEE COMSOC Best Survey Paper Award, the Best Scientific Contribution Award to the IEEE International Symposium on Wireless Systems 2017, and the Best Paper Award in Small Cell and 5G Networks (SmallNets) Workshop of the 2015 IEEE International Conference on Communications (ICC). He is the recipient of the IEEE ComSoc Outstanding Young Researcher Award for Europe, Middle East, $\&$ Africa Region in 2018. He is recognized as an exemplary reviewer by the IEEE Transactions on Communications for the three years 2014-2016, by the IEEE Transactions on Wireless Communications in 2017 and 2018, and by the IEEE Wireless Communications Letters in 2018. 
\end{IEEEbiography}

\begin{IEEEbiography}[{\includegraphics[width=1in,height=1.25in,clip,keepaspectratio]{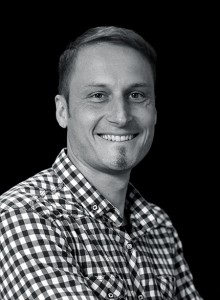}}]{Jacek Kibi\l{}da} received the M.Sc. degree from the Poznan University of Technology, Poznan, Poland, in
2008, and the Ph.D. degree from Trinity College, The University of Dublin, Dublin, Ireland, in 2016. He is currently a Research Fellow with CONNECT, Trinity College, The University of Dublin. His research interests include architectures and models for future mobile networks.
\end{IEEEbiography}

\begin{IEEEbiography}[{\includegraphics[width=1in,height=1.25in,clip,keepaspectratio]{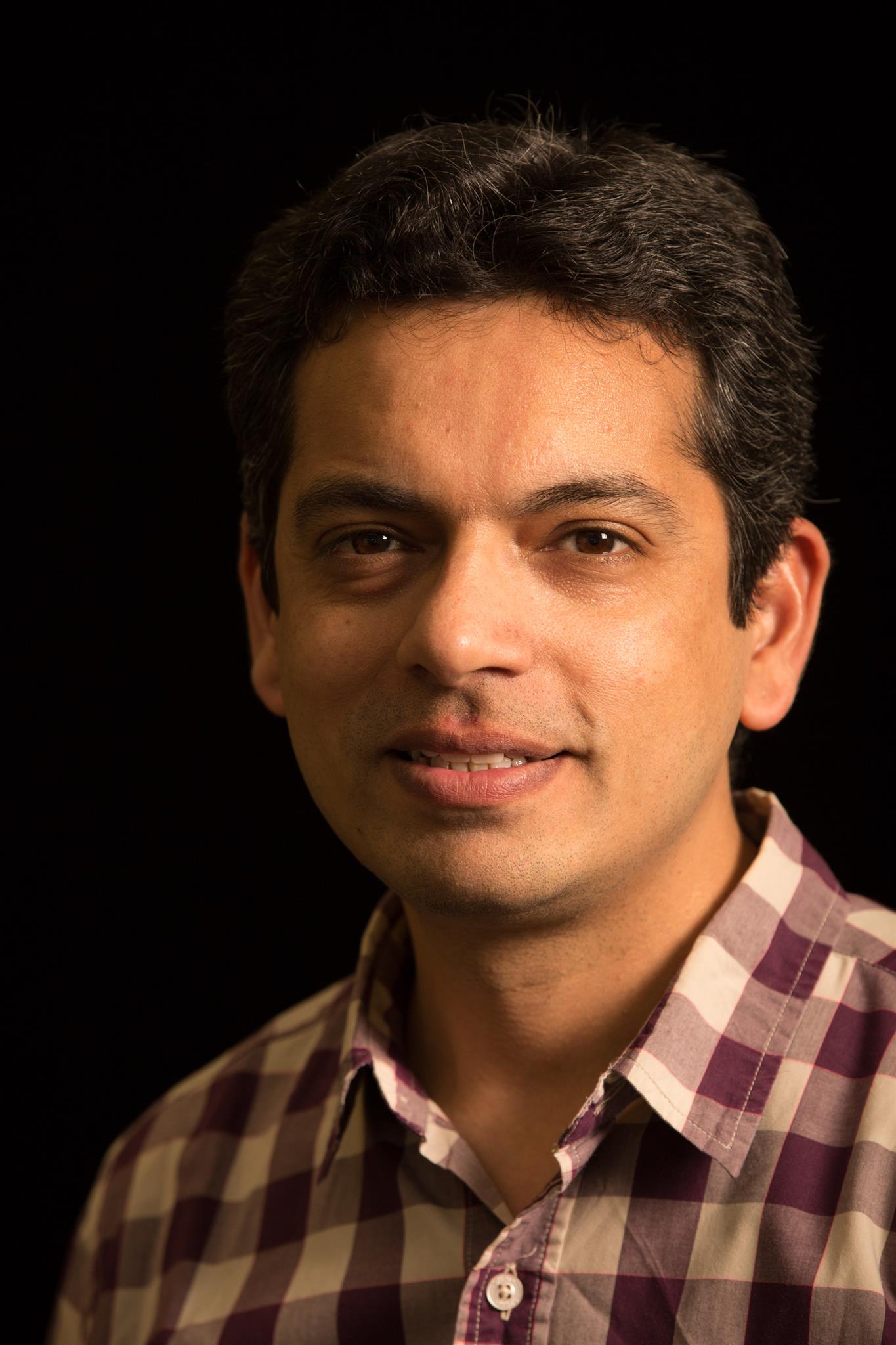}}]{M. Majid Butt} (S'07 -- M'10 -- SM'15) received the MSc degree in Digital Communications from Christian Albrechts University, Kiel, Germany, in 2005, and the PhD degree in  Telecommunications from the Norwegian University of Science and Technology, Trondheim, Norway, in 2011. He is an Assistant Professor at University of Glasgow as well as an adjunct Assistant Professor at Trinity College Dublin, Ireland. Before that, he has held senior researcher positions at Trinity College Dublin, Ireland and Qatar University. He is recipient of Marie Curie Alain Bensoussan postdoctoral fellowship from European Research Consortium for Informatics and Mathematics (ERCIM). He held ERCIM postdoc fellow positions at Fraunhofer Heinrich Hertz Institute, Germany, and University of Luxembourg. Dr. Majid's major areas of research interest include communication techniques for wireless networks with focus on radio resource allocation, scheduling algorithms, energy efficiency and cross layer design. He has authored more than 50 peer reviewed conference and journal publications in these areas. He has served as TPC chair for various communication workshops in conjunction with IEEE WCNC, ICUWB, CROWNCOM, IEEE Greencom and Globecom. He is a senior member of IEEE and serves as an associate editor for IEEE Access journal and IEEE Communication Magazine since 2016.
\end{IEEEbiography}

\newpage 

\begin{IEEEbiography}[{\includegraphics[width=1in,height=1.25in,clip,keepaspectratio]{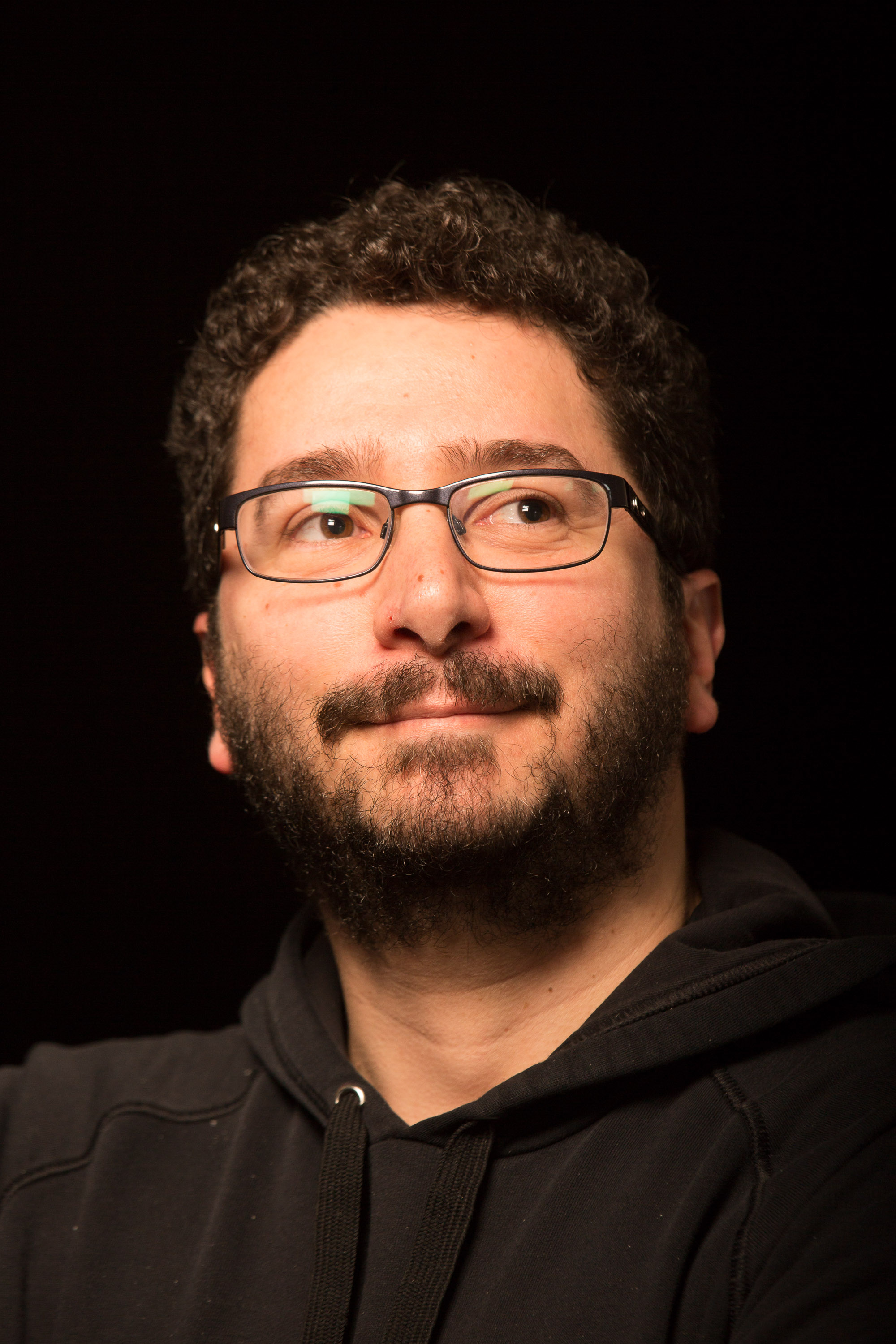}}]{Nicola Marchetti} Dr. Nicola Marchetti is currently Assistant Professor in Wireless Communications at Trinity College Dublin, Ireland. He performs his research under the Trinity Information and Complexity Labs (TRICKLE) and the Irish Research Centre for Future Networks and Communications (CONNECT). He received the PhD in Wireless Communications from Aalborg University, Denmark in 2007, and the M.Sc. in Electronic Engineering from University of Ferrara, Italy in 2003. He also holds an M.Sc. in Mathematics which he received from Aalborg University in 2010. His collaborations include research projects in cooperation with Nokia Bell Labs and US Air Force Office of Scientific Research, among others. His research interests include Adaptive and Self-Organizing Networks, Complex Systems Science for Communication Networks, PHY Layer, Radio Resource Management. He has authored 130 journals and conference papers, 2 books and 8 book chapters, holds 3 patents, and received 4 best paper awards.
\end{IEEEbiography}





\newpage
\vspace{-0.3 cm}
\begin{appendices}
\section{Proof of Lemma \ref{ch4:comp-interference}}			
\vspace{-0.0 cm}
\label{app:proof-comp-interference}
In the following, by saying $u \in \Phi_{cp}$, we mean that $\boldsymbol{y} \in \Phi_{cp}$, where $u=\lVert\boldsymbol{x}+\boldsymbol{y}\rVert$.
\begin{align}
\mathscr{L}_{I_{\rm out}}(t)  &= \mathbb{E} \Bigg[e^{-t\gamma_d \sum_{\Phi_p^{!}} \sum_{u \in \Phi_{cp}}  G_{u}  u^{-\alpha}} \Bigg] 
   \nonumber 
    \\
&= \mathbb{E}_{\Phi_p} \Bigg[\prod_{\Phi_p^{!}} \mathbb{E}_{\Phi_{cp},G_{u}} e^{-t\gamma_d \sum_{u \in \Phi_{cp}}  G_{u}  u^{-\alpha}} \Bigg] \\
\label{ch4:interference-back}
  &= \mathbb{E}_{\Phi_p} \Big[\prod_{\Phi_p^{!}} \mathbb{E}_{\Phi_{cp}} \prod_{u \in \Phi_{cp}} \mathbb{E}_{u,G_{u}}  e^{-t\gamma_d G_{u}  u^{-\alpha}} \Big] , 
   \end{align} 
where $G_{u}=G_{\boldsymbol{y}}$ for ease of exposition; from the Rayleigh fading assumption, we get
   \begin{align}
&\mathscr{L}_{I_{\rm out}}(t)\overset{}{=} \mathbb{E}_{\Phi_p} \Big[\prod_{\Phi_p^{!}} \mathbb{E}_{\Phi_{cp}} \prod_{u \in \Phi_{cp}} \mathbb{E}_{u}\frac{1}{1+t\gamma_d u^{-\alpha}} \Big] 
      \nonumber \\
          &\overset{(a)}{=}  \mathbb{E}_{\Phi_p} \Big[\prod_{\Phi_p^{!}} 
   {\rm exp}\Big(-p\bar{n} \int_{u=0}^{\infty}\Big(1 - \frac{1}{1+t\gamma_d u^{-\alpha}}\Big)f_U(u|v)\dd{u}\Big], 
\nonumber    
   \end{align} 
where (a) follows from the \ac{PGFL} of the Gaussian \ac{PPP} $\Phi_{cp}$. Notice that, in step (a), the \ac{PGFL} of the Gaussian \ac{PPP} $\Phi_{cp}$ is adopted for the intensity function given in polar coordinates rather than Cartesian coordinates, i.e., $v=\lVert\boldsymbol{x}\rVert$ and $u=\lVert\boldsymbol{x}+\boldsymbol{y}\rVert$.
%
Substituting $\int_{u=0}^{\infty}\big(1 -\frac{1}{1+t\gamma_d u^{-\alpha}}\big) f_{U|V}(u|v)\dd{u}= \int_{u=0}^{\infty}\frac{t\gamma_d}{u^{\alpha}+t\gamma_d} f_{U|V}(u|v)\dd{u} = \zeta(v,t)$, we get
\begin{align}
& \mathscr{L}_{I_{\rm out}}(t)\overset{}{=}  \mathbb{E}_{\Phi_p} \Big[\prod_{\Phi_p^{!}} 
   {\rm exp}\big(-p\bar{n} \zeta(v,t)\big)\Big] 
           \nonumber 	\\
           \label{LT_fpgl0}
            &\overset{(b)}{=}   {\rm exp}\Bigg( -2\pi \lambda_p\int_{v=0}^{\infty}\Big(1 - {\rm exp}\big(- p\bar{n}\zeta(v,t)\big)\Big)v\dd{v}\Bigg),
\end{align}  
where (b) follows from the \ac{PGFL} of the  \ac{PPP} $\Phi_p$. Hence, Lemma \ref{ch4:comp-interference} is proven.

\vspace{-0.3 cm}
\section{Proof of Theorem \ref{ch4:comp-interference-approx}}
\label{app:proof-comp-interference-approx}
By conditioning on $S_{\Phi_{cpm}}=s_{\Phi_{cpm}}=\sum_{i=1}^{k}  h_i^{-\alpha}$, we derive a bound on Laplace transform of inter-cluster interference based on Taylor's series expansion. Starting from equation (\ref{ch4:interference-back}) in \App{proof-comp-interference}, we have 		
\begin{align}
&\Lc_{I_{\rm out}}(t|k)= \mathbb{E}_{\Phi_p} \Bigg[\prod_{\Phi_p^{!}} \mathbb{E}_{\Phi_{cp}} \prod_{u \in \Phi_{cp}} \mathbb{E}_{u,G_u}  {\rm exp}\big(-t G_u  u^{-\alpha}\Big) \Bigg] \nonumber 	\\ 
 &\overset{(a)}{=}\Eb_{G_u} {\rm exp}\Big(-2\pi \lambda_p\int_{v=0}^{\infty}\big(1 - {\rm exp}\big(-p\bar{n}(1 - \zeta'(v,t))\big)\big)v\dd{v}\Big)	\nonumber	 \\ 
 &\overset{(b)}{\approx} \Eb_{G_u} {\rm exp}\Big(-2\pi \lambda_p\int_{v=0}^{\infty}\big(1 -  (1 - p\bar{n}(1 - \zeta'(v,t))\big)v\dd{v}\Big)	
\nonumber	\\
&={\rm exp}\Big(-2\pi p\bar{n}\lambda_p \overbrace{\big(\int_{v=0}^{\infty}v\dd{v} - \Eb_{G_u} \int_{v=0}^{\infty}\zeta'(v,t) v\dd{v}\big)\Big)}^{J(t)},
\nonumber 
\end{align}
where $\zeta'(v,t) = \int_{u=0}^{\infty}e^{-t\gamma_dG_uu^{-\alpha}} f_{U|V}(u|v)\dd{u}$ and $G_u= G_{\boldsymbol{y}}$ for ease of notation; (a) follows from tracking the proof of Lemma \ref{ch4:comp-interference} up until equation (\ref{LT_fpgl0}), (b) follows from Taylor series expansion for exponential function $e^{-x} \approx 1 - x$ when $x$ is small. It is worth mentioning that the obtained $\mathscr{L}_{I_{\rm out}}(t|k)$ in the above is Laplace transform of an upper bound on the interference. Correspondingly, the resulting rate coverage probability $\Upsilon_m$ and offloading gain $\mathbb{P}_{o}^{\sim}(\boldsymbol{c})$ are lower bounds on their exact values.  We proved in \cite[Lemma 2]{8647532} that  
$J(t)= \frac{(t\gamma_d)^{2/\alpha}}{2} \Gamma(1 + 2/\alpha)\Gamma(1 - 2/\alpha)$, which proves (\ref{eq:ppp-interference}). Plugging the result obtained \black{in Eq. (\ref{eq:ppp-interference})} into (\ref{eq:offload}) yields the lower bound on the offloading gain \black{in Eq. (\ref{ch4:lower-bound-offload-gain})}, which 
completes the proof. 

\vspace{-0.3 cm}
\section{Proof of Lemma \ref{ch4:pdf-nearest-sitance}}
\label{app:proof-pdf-nearest-distance}
With reference to Fig.~\ref{distance_near}, the nearest serving distance $h_1$ is defined as the distance from the typical client at $(0,0)$ to its nearest provider within the same cluster. Following \cite{7792210}, we define the \ac{PGF} of the number of active clients that cache content $m$ within a ball $\textbf{b}(o, h_1)$ with radius $h_1$ and centered around the origin $o$ as:
\begin{align}
&G_N(\vartheta)=  \Eb \left[ \vartheta^{\sum_{\boldsymbol{y}_{0i}\in\Phi_{cpm}}\textbf{1}\{\lVert \boldsymbol{x}_0 + \boldsymbol{y}_{0i}\rVert < h_1\}}\right]     
\nonumber \\
 &=  \Eb_{\Phi_c,\boldsymbol{x}_0} \prod_{\boldsymbol{y}_{0i}\in\Phi_{cpm}} \left[ \vartheta^{\textbf{1}\{\lVert \boldsymbol{x}_0 + \boldsymbol{y}_{0i}\rVert < h_1\}}\right]    \nonumber \\
&\overset{(a)}{=}  \Eb_{\boldsymbol{x}_0} {\rm exp}\Big(-c_mp\bar{n}\int_{\mathbb{R}^2}(1 - \vartheta^{\textbf{1}\{\lVert \boldsymbol{x}_0 + \boldsymbol{y}_{0i}\rVert < h_1\}})f_{\boldsymbol{Y}_{0i}}(\boldsymbol{y}_{0i})d\boldsymbol{y}_{0i}\Big) \nonumber \\
&\overset{(b)}{=}  \Eb_{\boldsymbol{x}_0} {\rm exp}\Big(-c_mp\bar{n}\int_{\mathbb{R}^2}(1 - \vartheta^{\textbf{1}\{\lVert \boldsymbol{z}_0 \rVert < h_1\}})f_{\boldsymbol{Y}_{0i}}(\boldsymbol{z}_0-\boldsymbol{x}_0)d\boldsymbol{z}_0\Big), 				
\nonumber 
 \end{align}
where $\textbf{1}\{.\}$ is the indicator function, and $\boldsymbol{x}_0\in \R^2$ is a \ac{RV} modeling the location of representative cluster center relative to the origin $o$, with a realization $\boldsymbol{X}_0=\boldsymbol{x}_0$;  (a) follows from the \ac{PGFL} of the \ac{PPP} $\Phi_{cpm}$ along with its intensity function $c_mp\bar{n} f_{\boldsymbol{Y}_{0i}}(\boldsymbol{y}_{0i})$, and (b) follows from change of variables $\boldsymbol{z}_0 = \boldsymbol{x}_0+\boldsymbol{y}_{0i}$. By converting Cartesian coordinates to polar coordinates with $h=\lVert \boldsymbol{z}_0 \rVert$, we get $G_N(\vartheta)=$
\begin{align}
&\Eb_{V_0} {\rm exp}\Big(-c_mp\bar{n}\int_{h=0}^{\infty}(1 - \vartheta^{\textbf{1}\{h < h_1\}})f_{H|V_0}(h|v_0)\dd{h}\Big) \nonumber \\
&\overset{(c)}{=} \Eb_{V_0} {\rm exp}\Big(-c_mp\bar{n}\int_{h=0}^{h_1}(1 - \vartheta)f_{H|V_0}(h|v_0)\dd{h}\Big) 
\overset{(d)}{=}
\nonumber \\
 &\int_{v_0=0}^{\infty}f_{V_0}(v_0)  {\rm exp}\Big(-c_mp\bar{n}\int_{h=0}^{h_1}(1 - \vartheta)f_{H|V_0}(h|v_0)\dd{h}\Big)\dd{v_0},
 \nonumber
 \end{align}
 where $V_0\in \R$ is a \ac{RV} modeling the distance from representative cluster's center to the origin $o$, with a realization $V_0=v_0=\lVert \boldsymbol{x}_0\rVert$; (c) follows from the definition of the indicator function 
$\textbf{1}\{h < h_1\}$, and (d) follows from unconditioing over $v_0$. 
To clarify how the normal distribution $f_{\boldsymbol{Y}_{0i}}(\boldsymbol{z}_0-\boldsymbol{x}_0)$ is converted to the Rician distribution $f_{H|V_0}(h|v_0)$, consider first the representative cluster centered at $\boldsymbol{x}_0 \in \Phi_p$, with a distance $v_0=\lVert \boldsymbol{x}_0\rVert$ from the origin. A randomly-selected active provider belonging to the representative cluster has its coordinates in $\mathbb{R}^2$ chosen independently from Gaussian distributions with standard deviation $\sigma$. Then, by definition, the distance $h$ from such an active provider to the origin has Rician \ac{PDF} denoted as $f_{H|V_0}(h|v_0)$. Recall that $f_{V_0}(v_0)=\mathrm{Rayleigh}(v_0,\sigma)$ from the definition of Gaussian \ac{PPP}.

Now, the \ac{CDF} of nearest serving distance $F_{H_1}(h_1)$ can be derived as
\begin{align}
\label{near-cdf}
& F_{H_1}(h_1)= 1 - G_N(0) = 1 -  \nonumber \\
&\int_{v_0=0}^{\infty} f_{V_0}(v_0) {\rm exp}\Big(-c_mp\bar{n}\int_{0}^{h_1}f_{H|V_0}(h|v_0)\dd{h}\Big)\dd{v_0}. 		
\end{align}
Applying Leibniz integral rule, we obtain the nearest distance \ac{PDF} as
\begin{align} 	
&f_{H_1}(h_1)=  -\frac{\partial}{\partial h_1} 
 \int_{v_0=0}^{\infty} f_{V_0}(v_0) 
e^{-c_mp\bar{n}\int_{0}^{h_1}f_{H|V_0}(h|v_0)\dd{h}}\dd{v_0}
\nonumber \\
&=  -\int_{v_0=0}^{\infty} f_{V_0}(v_0) \frac{\partial}{\partial h_1}  e^{-c_mp\bar{n}\int_{0}^{h_1}f_{H|V_0}(h|v_0)\dd{h}}\dd{v_0} = c_mp\bar{n} \times
\nonumber \\  
& \int_{0}^{\infty} f_{V_0}(v_0)  
\frac{\partial}{\partial h_1}\Big[\int_{0}^{h_1}f_{H|V_0}(h|v_0)\dd{h}\Big]
e^{-c_mp\bar{n}\int_{0}^{h_1}f_{H|V_0}(h|v_0)\dd{h}}\dd{v_0}
\nonumber \\
&= c_mp\bar{n} \int_{v_0=0}^{\infty} f_{V_0}(v_0) f_{H_1|V_0}(h_1|v_0) e^{-c_mp\bar{n}\int_{0}^{h_1}f_{H|V_0}(h|v_0)\dd{h}}\dd{v_0},
\nonumber
 \end{align}
The distance \ac{PDF} $f_{H_1}(h_1)$ can be calculated numerically from (\ref{Leibniz}). However, a tractable yet accurate approximation can be obtained using Jensen's inequality as follows: 
\begin{align}
&f_{H_1}(h_1)=\frac{\partial}{\partial h_1} F_{H_1}(h_1)
 \nonumber \\
 &\overset{(a)}{\approx} \frac{\partial}{\partial h_1}  \Bigg (1 - e^{-c_mp\bar{n}\int_{0}^{h_1}\int_{0}^{\infty} f_{V_0}(v_0) f_{H|V_0}(h|v_0)\dd{v_0}\dd{h}} \Bigg ) \nonumber \\
&\overset{(b)}{=} \frac{\partial}{\partial h_1}  \Bigg ( 1 - {\rm exp}\Big(-c_mp\bar{n}\big(1 -{\rm exp}(-\frac{h_1^2}{4\sigma^2}\big)\Big) \Bigg)
\label{approx_CDF}
 \\
&=  \frac{c_mp\bar{n} h_1 e^{-c_mp\bar{n} \big(1-e^{\frac{-h_1^2}{4 \sigma^2}}\big)-\frac{h_1^2}{4 \sigma^2}}}{2 \sigma^2},
 \end{align}
 where (a) follows from Jensen's inequality applied to the \ac{CDF} $F_{H_1}(h_1)$, and (b) follows from $\int_{0}^{h_1}\int_{v_0=0}^{\infty} f_{V_0}(v_0) f_{H|V_0}(h|v_0)\dd{v_0}\dd{h} = 1 - {\rm exp}(\frac{-h_1^2}{2\sigma^2})$. This completes the proof.

\vspace{-0.3 cm}
\section{Proof of Lemma \ref{ch4:variance-pwr}}
\label{app:proof-concentration}
The conditional variance ${\rm Var}\left[S_{\Phi_{cpm}^{!}}| H_1=h_1\right]$ can be expressed as
 \begin{align}
& {\rm Var}\left[S_{\Phi_{cpm}^{!}}|H_1=h_1 \right] = {\rm Var} \Big[ \sum_{\boldsymbol{y}_{0i}\in \Phi_{cpm}^{!}} \lVert \boldsymbol{x}_0 + \boldsymbol{y}_{0i}\rVert^{-\alpha} \Big] \nonumber \\
 &\overset{(a)}{=}	 \int_{\mathbb{R}^2}^\infty \frac{1}{\lVert \boldsymbol{x}_0 + \boldsymbol{y}_{0i}\rVert^{2\alpha}}c_mp\bar{n} f_{\boldsymbol{Y}_{0i}}(\boldsymbol{y}_{0i})\dd{\boldsymbol{y}_{0i}} 	
 \nonumber  \\
 &\overset{(b)}{=}	c_mp\bar{n} \int_{\mathbb{R}^2}^\infty \frac{1}{\lVert \boldsymbol{z}_0\rVert^{2\alpha}}	f_{\boldsymbol{Y}_{0i}}(\boldsymbol{z}_0 - \boldsymbol{x}_0)\dd{\boldsymbol{z}_0}, 
 \end{align}
where (a) follows from the mean and variance for \acp{PPP} \cite[Corollary 4.8] {haenggi2012stochastic}, along with the Gaussian \ac{PPP} assumption $\Phi_{cm}$; (b) follows from the substitution $\boldsymbol{z}_0 = \boldsymbol{x}_0 + \boldsymbol{y}_{0i}$, where $\{\boldsymbol{x}_0, \boldsymbol{y}_{0i},\boldsymbol{z}_0\} \in \R^2$.  
%
By converting the Cartesian coordinates coordinates to polar coordinates, where $h=\lVert \boldsymbol{z}_0 \rVert$, and unconditioning over $v_o$, we get
  \begin{align}
&{\rm Var}\left[S_{\Phi_{cpm}^{!}}| H_1=h_1\right]  =	c_mp\bar{n}  \int_{h_1}^\infty h^{-2\alpha}f_{H|V_0}(h|v_0)\dd{h}		\nonumber \\
 &\overset{}{=}	c_mp\bar{n}  \int_{v_0=0}^\infty f_{V_0}(v_0) \int_{h=h_1}^\infty h^{-2\alpha}f_{H|V_0}(h|v_0)\dd{h}\dd{v_0}   \nonumber 
\end{align}
\begin{align} 
  &\overset{(c)}{=}	c_mp\bar{n} \int_{h=h_1}^\infty h^{-2\alpha} \int_{v_0=0}^\infty f_{V_0}(v_0) f_{H|V_0}(h|v_0)\dd{v_0}\dd{h},	
 \end{align}
where (c) follows from changing the order of integration. Finally, we proceed as follows: ${\rm Var}\big[S_{\Phi_{cpm}^{!}}| H_1=h_1\big]=$  
\begin{align}
&c_mp\bar{n} \int_{h=h_1}^\infty h^{-2\alpha} 
\int_{0}^\infty \frac{v_0}{\sigma^2}e^{\frac{-v_0^2}{2\sigma^2}}
 \quad \frac{h}{\sigma^2}e^{{\frac{-(v_0^2+h^2)}{2\sigma^2}}}I_0(\frac{hv_0}{\sigma^2})\dd{v_0}\dd{h}			\nonumber \\
  &=	\frac{c_mp\bar{n}}{{\sigma^2}} \int_{h=h_1}^\infty h^{1-2\alpha} \int_{v_0=0}^\infty \frac{v_0}{\sigma^2}e^{\frac{-v_0^2}{2\sigma^2}} e^{{\frac{-(v_0^2+h^2)}{2\sigma^2}}}I_0(\frac{hv_0}{\sigma^2})\dd{v_0}\dd{h}			\nonumber \\
  \label{int-here1}
  & =	\frac{c_mp\bar{n}}{{\sigma^2}} \int_{h=h_1}^\infty h^{1-2\alpha} e^{\frac{-h^2}{2\sigma^2}}\int_{v_0=0}^\infty \frac{v_0}{\sigma^2}e^{\frac{-v_0^2}{\sigma^2}} I_0(\frac{hv_0}{\sigma^2})\dd{v_0}\dd{h}			
   \\
\label{int-here2}
&\quad\quad\quad  \overset{(d)}{=}	\frac{c_mp\bar{n}}{2\sigma^2} \int_{h_1}^\infty h^{1-2\alpha} e^{\frac{-h^2}{4\sigma^2}} \dd{h} 
\overset{(e)}{=} c_mp\bar{n} \int_{\frac{h^2_1}{4\sigma^2}}^\infty \tau^{-2\alpha} e^{-\tau} \dd{\tau}
\\
&\quad\quad\quad \overset{(f)}{=} c_mp\bar{n} \Gamma\left(-2\alpha+1,\frac{h^2_1}{4\sigma^2}\right),
\end{align}
where (d) follows from solving the inner integratal of (\ref{int-here1}), (e) follows from the substitution $\tau=\frac{h^2}{4\sigma^2}$, and (f) follows from solving the integration of (\ref{int-here2}), where $\Gamma(\cdot,\cdot)$ denotes the upper incomplete gamma function. This completes the proof.

\end{appendices}

\end{document}